%% file: man.tex
\pdfoutput=1
\documentclass{article}
\usepackage[utf8]{inputenc}
\usepackage[T1]{fontenc}
\usepackage{graphicx}
\usepackage{grffile}
\usepackage{longtable}
\usepackage{wrapfig}
\usepackage{rotating}
\usepackage[normalem]{ulem}
\usepackage{amsmath}
\usepackage{textcomp}
\usepackage{amssymb}
\usepackage{capt-of}
\usepackage{hyperref}
\input{shared}
\date{November 6, 2017}
\title{Lurking Variable Detection via Dimensional Analysis}
\hypersetup{
 pdfauthor={Squire},
 pdftitle={},
 pdfkeywords={},
 pdfsubject={},
 pdfcreator={Emacs 25.3.2 (Org mode 9.0.9)},
 pdflang={English}}
\begin{document}

\maketitle

\begin{abstract}
Lurking variables represent hidden information, and preclude a
full understanding of phenomena of interest. Detection is usually
based on serendipity -- visual detection of unexplained,
systematic variation. However, these approaches are doomed to
fail if the lurking variables do not vary. In this article, we
address these challenges by introducing formal hypothesis tests
for the presence of lurking variables, based on Dimensional
Analysis. These procedures utilize a modified form of the
Buckingham $\pi$ theorem to provide structure for a suitable null
hypothesis. We present analytic tools for reasoning about lurking
variables in physical phenomena, construct procedures to handle
cases of increasing complexity, and present examples of their
application to engineering problems. The results of this work
enable algorithm-driven lurking variable detection, complementing
a traditionally inspection-based approach.
\end{abstract}

\section{Introduction}
\label{sec:org1320493}
Understanding the relationship between inputs (variables) and outputs
(responses) is of critical importance in uncertainty quantification (UQ).
Relationships of this sort guide engineering practice, and are used to optimize
and certify designs. UQ seeks to quantify variability, whether in a forward or
inverse sense, in order to enable informed decision making and mitigate or
reduce uncertainty. The first step in this practice is to choose the salient
responses or \emph{quantities of interest}, and identify the variables which affect
these quantities. However, efforts to quantify and control uncertainty will be
thwarted if critical variables are neglected in the analysis -- if \emph{lurking
variables} affect a quantity of interest.

Lurking variables are, by definition, confounding. Examples include unnoticed
drift of alignment in laser measurements, unmeasured geometry differences
between articles in a wind tunnel, and uncontrolled temperature fluctuations
during materials testing. Such variables are called lurking if they affect a
studied response, but are unaccounted for in the analysis. More explicitly,
there exist \(\nvar\) variables which affect some quantity of interest; a subset
of these are known, while the remainder are said to be lurking
variables.\cite{box1966} Such lurking variables represent a pernicious sort of
uncertainty -- a kind of unknown unknown. Whether studying a closed-source
black-box simulation code or a physical experiment, the presence of unaccounted
factors will stymie the development of understanding.

As a historical example, consider the 1883 work of Osborne Reynolds
\cite{reynolds1883} on the flow of a viscous fluid through a pipe. Reynolds
studied the effects of the fluid density, viscosity, bulk velocity, and pipe
diameter on pressure losses. He found that for certain physical regimes, these
variables were insufficient to describe the observed variability. Through
domain-specific knowledge, Reynolds was able to deduce that the surface
roughness of the pipes accounted for the unexplained behavior, a deduction that
was confirmed and thoroughly studied in later works.\cite{nikuradse1950}

While the lurking variable issue Reynolds encountered generalizes to other
settings, his solution technique certainly does not. Ideally, one would like a
strategy for identifying when lurking variables affect a studied response. While
the issue of lurking variables has received less attention in the Uncertainty
Quantification community, Statisticians have been grappling with this issue for
decades: Joiner \cite{joiner1981} recommends checking ``a variety of plots of the
data and the residuals'' as tools to combat lurking variables. Such graphical
inspection-based approaches are foundationally important. However, in the case
where a lurking variable is \textbf{fixed in value} for the duration of the experiment,
even the most careful graphical inspection (of this sort) is doomed to fail in
detecting it. An analyst would ideally like to be able to check whether the
observed relationship between response and predictors is consistent with the a
priori assumption of no lurking variables. In general, no structure exists to
frame such a hypothesis. However, in the context of a physical experiment,
Dimensional Analysis provides the means to impose such structure.

This structure is provided by the Buckingham \(\pi\) theorem, a consequence of
dimensional homogeneity that applies to all physical
systems.\cite[Ch. 0]{barenblatt1996} Dimensional Analysis begins with a priori
information about the physical system -- the variables' physical dimensions --
and imposes a constraint on the allowable functional form relating predictors
and response. Palmer \cite{palmer2008} colorfully refers to Dimensional Analysis
as a means ``to get something for nothing''. However, Dimensional Analysis
hinges on correct information; Albrecht et al. \cite{albrecht2013} write in the
context of experimental design that ``the scientist must know, a priori, the
complete set of independent variables describing the behavior of a system. If
the independent variables are misspecified (e.g., if one variable is missing),
the results of the Dimensional Analysis experiment may be completely unusable.''
It is this failure mode we aim to address. The key insight of the present work
is to leverage the Buckingham \(\pi\) theorem as testable structure.

In this work, we present a procedure for null-hypothesis significance testing of
the presence of lurking variables. This procedure is based on a formal truth
model derived from Dimensional Analysis, which includes all relevant factors.
This idea has been pursued in other works; Pearl and Bareinboim \cite{pearl2014}
present nonparametric structural equations models that incorporate so-called
`exogenous variables', with an eye towards generality. We restrict attention to
dimensional lurking variables, in order to impose testable structure and develop
a detection procedure. Previous work has explored model choices which respect
dimensional homogeneity, such as the additive power-law model of Shen and Lin
\cite{shen2017conjugate}. We avoid specific model choices, and instead work with
the fundamental properties of dimensionally homogeneous relationships. Shen and
Lin highlight an important advantage of Dimensional Analysis; namely, its
ability to provide meaningful insight into physical variables, even ones which
are fixed in value. They note ``with the help of (Dimensional Analysis), the
effect of these physical constants can be automatically discovered and
incorporated into the results without actually varying them.'' It is this
property which enables our procedure to avoid reliance on serendipity, and
detect lurking variables which are fixed in value. Furthermore, below we develop
the analysis and methods to \emph{choose} to fix a predictor, and perform detection
in the face of such \emph{pinned variables}.

An outline of this article is as follows. Section \ref{sec:org6f9d706} provides
an overview of Dimensional Analysis and derives a testable form of the
Buckingham \(\pi\) theorem for lurking variable detection. Section \ref{sec:orgb99b5a2} reviews Stein's lemma and Hotelling's \(T^2\) test, which form
the statistical basis for our procedures. Section \ref{sec:org26bf4ba}
combines Dimensional Analysis with hypothesis testing for a lurking variable
detection procedure, and provides some guidance on sampling design and power
considerations. Section \ref{sec:orgbc56bcc} demonstrates this procedure on a
number of engineering-motivated problems, while Section \ref{sec:org0b49157} provides
concluding remarks. The results of this paper are to enable \emph{algorithm-driven}
detection of lurking variables, capabilities which lean not on expert experience
or serendipity, but rather on an automated, data-driven procedure.

\section{Dimensional Analysis}
\label{sec:org6f9d706}
  Dimensional Analysis is a fundamental idea from physics. It is based on a
simple observation; the physical world is indifferent to the arbitrary unit
system we define to measure it. A \emph{unit} is an agreed-upon standard by which we
measure physical quantities. In the International System of units, there are
seven such \emph{base units}: the meter \((m)\), kilogram \((kg)\), second \((s)\), ampere
\((A)\), kelvin \((K)\), mole \((mol)\) and candela \((cd)\).\cite{thompson2008} All other
\emph{derived units} may be expressed in terms of these seven quantities; e.g. the
Newton \(kg\cdot m/s^2\). The choice of base units is itself arbitrary, and
constitutes a choice of a class of unit system.

Contrast these standard-defined units with \emph{dimensions}. Units such as the
meter, kilogram, and second are standards by which we measure length (\(L\)), mass
(\(M\)), and time (\(T\)). Barenblatt \cite[Ch. 1]{barenblatt1996} formally
defines the \emph{dimension function} or \emph{dimension} as ``the function that
determines the factor by which the numerical value of a physical quantity
changes upon passage from the original system of units to another system within
a given class''. Engineers commonly denote this dimension function by square
brackets; e.g. for a force \(F\), we have \([F]=M^1L^1T^{-2}\). For such a quantity,
rescaling the time unit by a factor \(c\) rescales \(F\) by a factor \(c^{-2}\). Note
that the dimension function is required to be a power-product, following from
the principle of absolute significance of relative
magnitude.\cite[Ch. 2]{bridgman1922dimensional} A quantity with non-unity
dimension is called a \emph{dimensional quantity}, while a quantity with dimension
unity is called a \emph{dimensionless quantity}. To avoid confusion with dimension in
the sense of number of variables, we will use the term \emph{physical dimension} for
the dimension function.

Note that dimensional quantities are subject to change if one modifies their
unit system (See Remark \ref{rmk:units}). Since the physical world is invariant
to such capricious variation, it must instead depend upon dimensionless
quantities, which are invariant to changes of unit system. The formal version of
this statement is the Buckingham \(\pi\) theorem.\cite{buckingham1914}

\subsection{The Buckingham \(\pi\) theorem}
\label{sec:org700160e}
In what follows, we use unbolded symbols for scalars, bolded lowercase symbols
for vectors, and bolded uppercase letters for matrices. Given a dimensionless
quantity of interest (qoi) \(\pi=f(\v\var)\in\R{}\) affected by factors
\(\v\var\in\R{\nvar}\) with \(\rnk\geq0\) independent physical dimensions (to be
precisely defined below) among a total of \(\ndim\geq\rnk\) physical dimensions,
this physical relationship may be re-expressed as

\begin{equation}
  \begin{aligned}
    \pi &= f(\v\var), \\
    &= \psi(\pi_1,\dots,\pi_{\nvar-\rnk}),
  \end{aligned}
\end{equation}

\noindent where the \(\pi_i\) are independent dimensionless quantities, and \(\psi\)
is a new function of these \(\pi_i\) alone. This leads to a simplification through
a reduction in the number of variables. Some physical quantities are inherently
dimensionless, such as angles.\cite[Table 3]{thompson2008} However, most
dimensionless quantities \(\pi_i\) are formed as combinations of dimensional
quantities \(\var_j\) in the form of a power-product

\begin{equation}
  \pi_i = \prod_{j=1}^p \var_j^{v_{ij}},
\end{equation}

\noindent where \(\vv_i\in\R{\nvar}\). The elements of \(\vv_i\) are chosen such that
\([\pi_i]=1\).

Valid dimensionless quantities are defined by the nullspace of a particular
matrix, described here. Let \(\mD\in\R{d\times p}\) be the \emph{dimension matrix} for
the \(p\) physical inputs \(\v\var=(z_1,\dots,z_p)^T\) and \(d\) physical dimensions
necessary to describe them. Note that in the SI system, we have
\(d\leq7\).\cite{thompson2008} The introduction of \(\mD\) allows for a linear algebra
definition of \(r\); we have \(r=\text{Rank}(\mD)\), and so \(r\leq d\). Similarly,
the \(\pi_i\) are independent in a linear algebra sense, with the dimensionless
quantities defined by their respective vectors \(\vv_i\).

The columns of \(\mD\) define the physical dimensions for the associated variable.
For example, if \(\var_1=\rho_F\) is a fluid density with physical dimensions
\([\rho_F]=M^1L^{-3}\), and we are working with \(\ndim=3\) dimensions \((M,L,T)\),
then \(\rho_F\)'s corresponding column in \(\mD\) will be \(\vd_1=(+1,-3,+0)^T\). For
convenience, we introduce the \emph{dimension vector operator} \(\vd(\cdot)\), which
returns the vector of physical dimension exponents; e.g. \(\vd(\rho_F)=\vd_1\)
above. Note that the dimension vector operator is only defined with respect to a
dimension matrix. The nullspace of \(\mD\) defines the dimensionless quantities of
the physical system.

One may regard vectors in the domain (\(\R{\nvar}\)) of \(\mD\) as defining products
of input quantities; suppose we expand our example to consider the inputs for
Reynolds' problem of Rough Pipe Flow, consisting of a fluid density \(\rho_F\),
viscosity \(\mu_F\), and (bulk) velocity \(U_F\), flowing through a pipe with
diameter \(d_P\) and roughness lengthscale\footnote{Roughness lengthscale is a measure
of surface roughness of the interior of a pipe; it is often considered a
material property.\cite{nikuradse1950}} \(\epsilon_P\). We then have the dimension
matrix considered in Table \ref{tab:pipe_dimensions}.

\begin{table}[!ht]
  \centering
  \begin{tabular}{@{}lccccc@{}}
    \toprule
    Dimension  & $\rho_F$ & $U_F$   & $d_P$ &  $\mu_F$ & $\epsilon_P$\\
    \midrule
    Mass (M)   &   1    &  0    &  0  &    1   &  0 \\
    Length (L) & \mm3   &  1    &  1  &  \mm1  &  1 \\
    Time (T)   &   0    & \mm1  &  0  &  \mm1  &  0 \\
    \bottomrule
  \end{tabular}
  \caption{Dimension matrix for Rough Pipe Flow.}
  \label{tab:pipe_dimensions}
\end{table}

\noindent The vector \(\vv\equiv(+1,+1,+1,-1,+0)^T\) lies in the nullspace of
\(\mD\) for Rough Pipe Flow, and may be understood as the powers involved in the
product \(Re=\rho_F^1U_F^1d_P^1\mu_F^{-1}\epsilon_P^0\); this is a form of the
Reynolds number, a classic dimensionless quantity from fluid
mechanics.\cite{white2011} We elaborate on this example below.

\subsection{Illustrative example}
\label{sec:org70885ee}
In this section, we consider an example application of Dimensional Analysis to
illustrate both its application, and the class of problem our procedures are
designed to solve. We consider the physical problem of Rough Pipe Flow; that is,
the flow of a viscous fluid through a rough pipe, visually depicted in Figure
\ref{fig:reynolds}.

\begin{figure}[!ht]
\begin{minipage}{0.49\textwidth}
\includegraphics[width=0.95\textwidth]{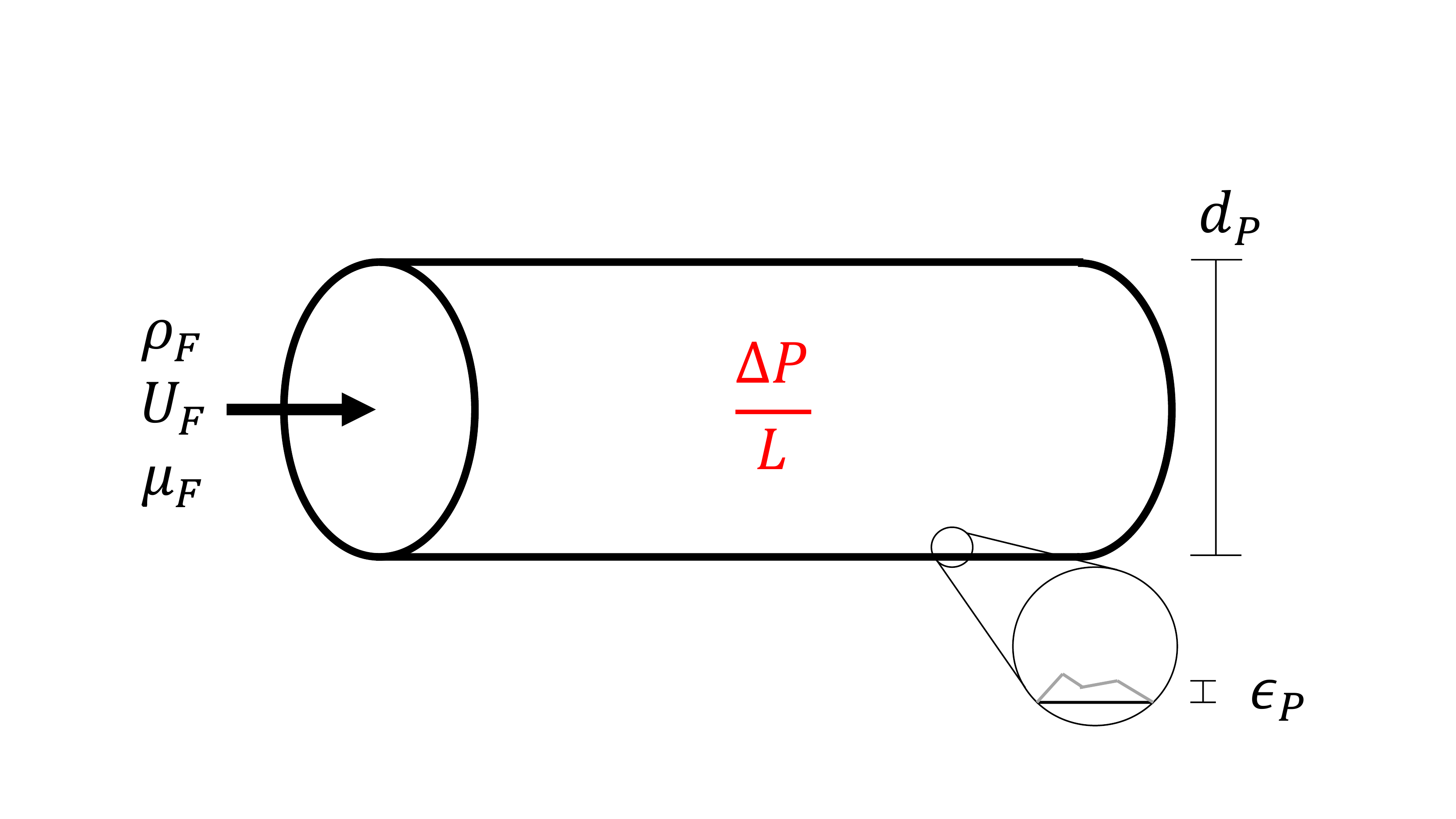}
\end{minipage} %
\begin{minipage}{0.49\textwidth}
\includegraphics[width=0.95\textwidth]{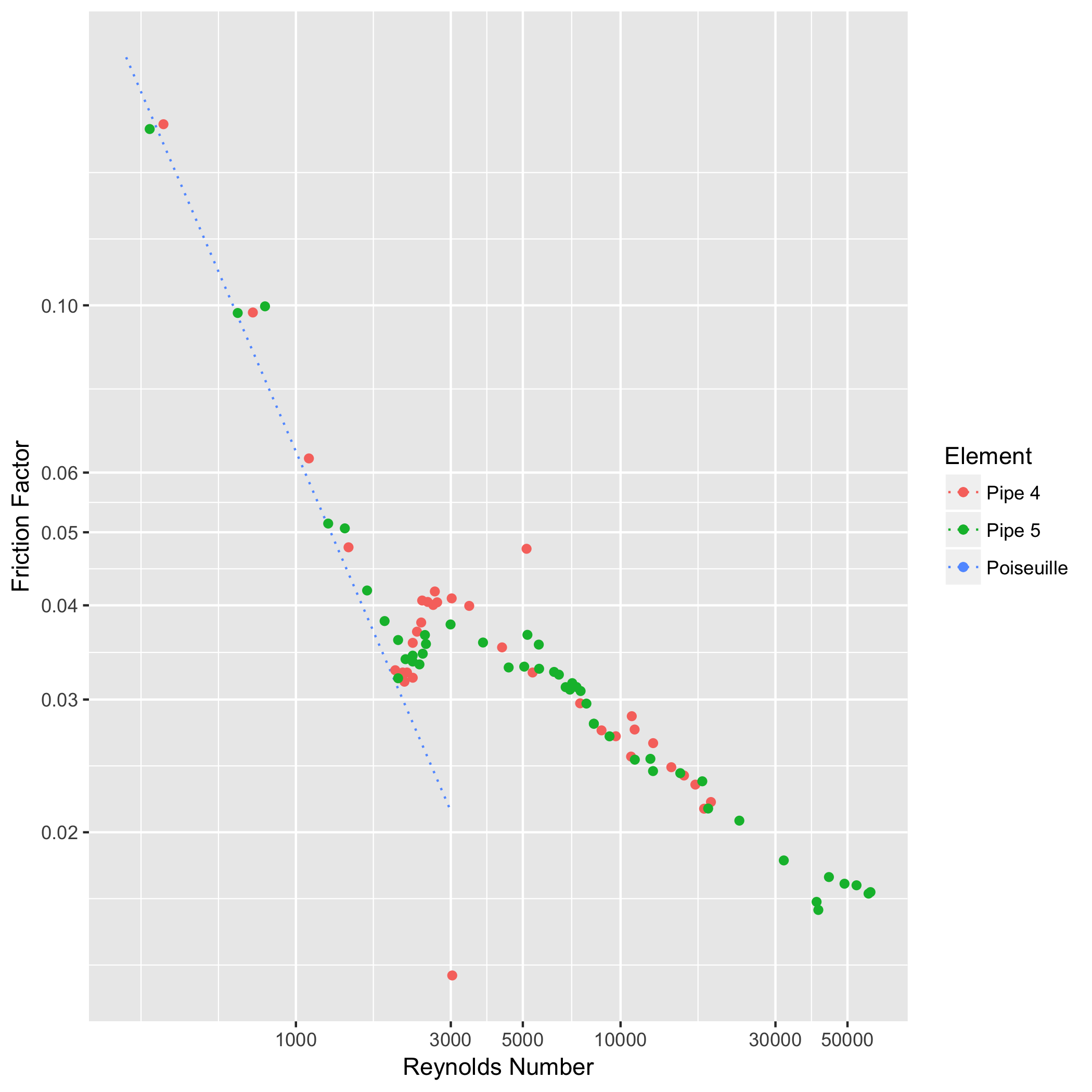}
\end{minipage}
\caption{Schematic for Rough Pipe Flow (left) and Reynolds original 1883 data (right).
  Osborne Reynolds constructed pipes of varying diameter and modified
  the fluid velocity and temperature conditions (affecting viscosity); his
  original data is non-dimensionalized and plotted against the
  Dimensional Analysis-predicted dimensionless variable in the right panel.
  The data approximately collapse to a one-dimensional function of the Reynolds
  number, demonstrating the simplifying power of Dimensional Analysis.
  Note that in laminar conditions ($Re\lessapprox3000$), the Poiseuille
  equation is a good model for the response, while behavior changes dramatically
  in turbulent conditions.\cite{white2011} After the
  turbulent onset, the relative roughness $\epsilon_P/d_P$ affects the response,
  which Reynolds observed not in his data, but in that of Henry
  Darcy.\cite{nikuradse1950}}
\label{fig:reynolds}
\end{figure}

In his seminal paper, Osborne Reynolds \cite{reynolds1883} considered the effects
of the fluid bulk velocity \(U_F\), density \(\rho_F\), and viscosity \(\mu_F\), as
well as the pipe diameter \(d_P\) on qualitative and quantitative behavior of the
resulting flow. To apply Dimensional Analysis, one must first write down the
physical dimensions of the variables, as shown in Table
\ref{tab:pipe_dimensions_lurk}. This dimension matrix has a one-dimensional
nullspace, for which the vector \((1,1,1,-1)^T\) is a basis. The corresponding
dimensionless quantity is \(\rho_FU_Fd_F/\mu_F\), which was first considered
by Reynolds in his 1883 work.

\begin{table}[!ht]
  \centering
  \begin{tabular}{@{}lccccc@{}}
    \toprule
    Dimension  & $\rho_F$ & $U_F$   & $d_P$ &  $\mu_F$ \\
    \midrule
    Mass (M)   &   1    &  0    &  0  &    1   \\
    Length (L) & \mm3   &  1    &  1  &  \mm1  \\
    Time (T)   &   0    & \mm1  &  0  &  \mm1 \\
    \bottomrule
  \end{tabular}
  \caption{Dimension matrix for Rough Pipe Flow, neglecting roughness.}
  \label{tab:pipe_dimensions_lurk}
\end{table}

Reynolds found that his own data collapsed to a one dimensional function against
this dimensionless quantity (Fig. \ref{fig:reynolds}), and so applied his
findings to the data of Henry Darcy. Darcy had considered a number of pipes of
different materials, and Reynolds found he needed to apply a correction term in
order to bring the different datasets into agreement. Reynolds noted

``Darcy's pipes were all of them uneven between the gauge points, the glass and
the iron varying as much as 20 per cent.(sic) in section. The lead were by far
the most uniform, so that it is not impossible that the differences in the
values of \(n\) may be due to this unevenness. But the number of joins and
unevenness of the tarred pipes corresponded very nearly with the new cast iron,
and between these there is a very decided difference in the value of \(n\). This
must be attributed to the roughness of the cast iron surface.''

Reynolds correctly deduced the significance of the pipe roughness using prior
experience, but did not include a variable to represent it in his 1883 work. One
may regard roughness \(\epsilon_P\) as a lurking variable in Reynolds' original
setting. Reynolds manufactured two pipes (Pipe No. 4 and Pipe No. 5) of varying
diameter but the same material, varying \(d_P\) but leaving \(\epsilon_P\) fixed.
This feature of the experiment would have rendered lurking variable detection
via the usual approaches impossible. Could an experimentalist have detected the
lurking roughness through a statistical approach? We will provide evidence of an
affirmative answer through the examples below, but first provide some intuition
for how such a procedure is possible.

Note that considering all the salient factors for Rough Pipe Flow yields the
dimension matrix shown in Table \ref{tab:pipe_dimensions}. This matrix has a
two-dimensional nullspace, for which an acceptable basis is
\(\{(1,1,1,-1,0)^T,(0,-1,0,0,1)^T\}\). These vectors correspond respectively to
the Reynolds number and the \emph{relative roughness}
\(\epsilon_P/d_P\).\cite{nikuradse1950} Note that in Reynolds' setting, the observed
variables are \(\{\rho_F,U_F,d_P,\mu_F\}\) while the sole lurking variable is
\(\epsilon_P\). In this case there is only one predicted dimensionless quantity,
and so naive Dimensional Analysis suggests the response is one-dimensional in
\(Re\). However, the true physical relationship depends additionally on the
relative roughness \(\epsilon_P/d_P\); this results in additional variability due
to \(d_P\), and in two-dimensional behavior unpredicted by Dimensional Analysis on
the observed variables alone. One could attempt to perform visual inspection
based on the predicted dimensionless quantities; however, we will instead build
a formal hypothesis test. Below, we present a modified version of Buckingham
\(\pi\), first shown by Constantine et al.\cite{constantine2016} This provides a
\emph{sufficient dimension reduction} based on a priori knowledge, contingent on the
absence of lurking variables.\cite{adragni2009} This is the key to our lurking
variable detection strategy.
\subsection{The Pi Subspace and lurking variables}
\label{sec:org644fea1}
   The effect of the \(\log\) operator is to turn power-products into linear
combinations. If we define \(\v\tar\equiv\log(\v\var)\in\R{\nvar}\), with the
\(\log\) (and \(\exp\) to follow) interpreted in the elementwise sense, we may
rewrite the Buckingham \(\pi\) theorem as

\begin{equation}
  \begin{aligned}
    \pi &= f(\v\var), \\
    &= \psi(\pi_1,\dots,\pi_{\nvar-\rnk}), \\
    &= \psi'(\mV^T\v\tar),
  \end{aligned}
\end{equation}

\noindent where \(\psi'(\v\xi) = \psi(\exp(\v\xi))\), and
\(\mV\in\R{\nvar\times(\nvar-\rnk)}\) is a basis for the nullspace of \(\mD\). Let
\(\cR(\mA)\) denote the columnspace of a matrix \(\mA\). The subspace \(\cR(\mV)\) is
known from \(\mD\); if all the physical inputs \(\v\var\) have been correctly
identified, then \(\mD\) is known a priori. Suppose that \(\pi\) is measured via
some noisy instrument. Then \(\pi=\psi'(\mV^T\v\tar)+\epsilon'\). If we assume
\(\epsilon'\) to be unbiased, we have

\begin{equation}
  \pi | \v\tar \sim \pi | \mV^T\v\tar;
\end{equation}

\noindent that is, \(\cR(\mV)\) is a sufficient dimension
reduction.\cite{adragni2009} Intuitively, a subspace is sufficient if it captures
all available information about a qoi; moving orthogonal to \(\mV\) results in no
change to \(\pi\) outside random fluctuations. Since \(\cR(\mV)\) is determined from
the Buckingham \(\pi\) theorem, we call it the \emph{pi subspace}. Note that while the
pi subspace is uniquely determined by \(\mD\), the matrix \(\mV\) must be chosen.
Selecting an appropriate \(\mV\) is certainly of interest; both in engineering
\cite{delrosario2017} and statistical circles \cite{shen2017conjugate}. However, in
what follows we need only \(\cR(\mV)\), so the precise choice of \(\mV\) is not an
issue we consider in this work.

Of greater import in this discussion is the choice of \(\v\tar\) (equivalently
\(\v\var\)). An analyst chooses the relevant physical quantities based on previous
experience or intuition. Previous experience may fail to generalize, and intuition
can be faulty. In these cases, even an experienced investigator may fail to
identify all the relevant factors, and may instead consider \(\nvar_{\EX}<\nvar\)
observed or \emph{exposed} variables \(\v\tar_{\EX}\in\R{\nvar_{\EX}}\), leaving out
\(\nvar-\nvar_{\EX}\) lurking variables \(\v\tar_{\LU}\in\R{\nvar-\nvar_{\EX}}\).
Note that the ordering of variables in \(\vx\) is arbitrary; we assume an ordering
and write

\begin{equation}\label{eq:split}
\v\tar^T = (\v\tar_{\EX}^T,\v\tar_{\LU}^T).
\end{equation}

\noindent The analyst varies \(\v\tar_{\EX}\) in order to learn about the
functional dependence of the qoi on these exposed variables. In practice, the
analyst is aware of \(\v\tar_{\EX}\) and their physical dimensions
\(\mD_{\EX}\in\R{\ndim\times \nvar_{\EX}}\), but is totally ignorant of
\(\v\tar_{\LU}\) and their physical dimensions \(\mD_{\LU}\in\R{\ndim\times
(\nvar-\nvar_{\EX})}\). The full dimension matrix is then
\(\mD=[\mD_{\EX},\mD_{\LU}]\). In such a setting, the analyst would derive an
incorrect pi subspace corresponding to
\(\mV'_{\EX}\in\R{\nvar_{\EX}\times(\nvar_{\EX}-\rnk_{\EX})}\), where
\(\mD_{\EX}\mV'_{\EX}=0\). This is a recognized issue in the study of Dimensional
Analysis.\cite{sonin2001} In the case where
\(\cR(\mV'_{\EX})\not\supseteq\cR(\mV)\), the derived subspace will not be
sufficient. Note that the Buckingham \(\pi\) theorem constrains the gradient of
our qoi to live within the pi subspace; we may use this fact to provide
diagnostic information.

Note that

\begin{equation}
\nabla_{\v\tar}^T\pi = (\nabla_{\v\tar_{\EX}}^T\pi,\nabla_{\v\tar_{\LU}}^T\pi).
\end{equation}

\noindent Since \(\mD\nabla_{\v\tar}\pi=0\), we have

\begin{equation}\label{eq:dim-match}
  \mD_{\EX}\nabla_{\v\tar_{\EX}}\pi = -\mD_{\LU}\nabla_{\v\tar_{\LU}}\pi.
\end{equation}

\noindent \Cref{eq:dim-match} demonstrates that the vector of physical
dimensions of the exposed variables matches that of the lurking variables.
Furthermore, the left-hand side of \eqref{eq:dim-match} is composed of
known (or estimable) quantities. If \(\mD_{\EX}\nabla_{\v\tar_{\EX}}\pi\) is
nonzero, it signals that 1) a lurking variable exists, and 2) it possesses
physical dimensions aligned with \(\mD_{\LU}\nabla_{\v\tar_{\LU}}\pi\). In the
Rough Pipe Flow example, the entries of \(\mD_{\EX}\nabla_{\v\tar_{\EX}}\pi\)
correspond to Mass, Length, and Time, respectively. Thus if
\(\mD_{\EX}\nabla_{\v\tar_{\EX}}\pi=(0,c,0^T)\), a lurking lengthscale affects the
qoi. The information provided is richer than a simple binary detection; we can
gleam some physical information about the lurking variable. To capitalize on
this observation, we will base our detection procedure on estimating the
gradient of our quantity of interest.

There are important caveats to note. We may have \(\mD_{\LU}=0\) (i.e. the
\(\v\var_{\LU}\) are dimensionless), or we may have dimensional inputs but
\(\nabla_{\v\tar_{\LU}}\pi\neq0\) with \(\mD_{\LU}\nabla_{\v\tar_{\LU}}\pi=0\) (i.e.
the lurking variables form a dimensionless quantity). The latter case is
unlikely; this occurs when the analyst truly does not know much about the
problem at hand. The former case is more challenging -- natural physical
quantities exist that are inherently dimensionless, such as angles. To combat
these issues, an analyst may choose to decompose such dimensionless quantities
in terms of other dimensional ones; an angle may be considered a measure between
two vectors, so an analyst may introduce lengthscales that form the desired
angle. However, this requires intimate understanding of the failure mode, and
does not address the fundamental issue. In the case where such an unknown is
dimensionless, detection must be based on some other principle, as Dimensional
Analysis will not be of help.

\subsection{Non-dimensionalization}
\label{sec:org305f95f}
   Rarely is a measured qoi dimensionless; usually, a dimensionless qoi must be
derived via \emph{non-dimensionalization}. Given some dimensional qoi \(\qoi\), we form
the \emph{non-dimensionalizing factor} via a product of the physical input factors
\(\exp(\vu^T\v\tar)=\prod_{i=1}^{\nvar}\var_i^{u_i}\), such that
\(\vd(\qoi)=\vd(\prod_{i=1}^{\nvar}\var_i^{u_i})\). We then form the dimensionless
qoi via

\begin{equation}
  \begin{aligned}
    \pi(\v\tar) &= \qoi(\v\tar) \prod_{i=1}^{\nvar}\var_i^{-u_i}, \\
        &= \qoi(\v\tar) \exp(-\vu^T\v\tar).
  \end{aligned}
\end{equation}

\noindent One may determine an acceptable non-dimensionalizing factor by solving
the linear system

\begin{equation}\label{eq:dim-match-vec}
  \vd(\qoi)=\mD\vu.
\end{equation}

\noindent If \eqref{eq:dim-match-vec} possesses no solution, then no
dimensionally homogeneous relationship among the qoi and proposed variables
exists. Since the physical world is required to be dimensionally homogeneous, we
conclude that lurking variables must affect the qoi. Bridgman
\cite[Ch. 1]{bridgman1922dimensional} illustrates this point through various
examples.

More commonly, \Cref{eq:dim-match-vec} will possess infinite solutions; however,
under sensible conditions (Sec. \ref{sec:org8e1bcfd}), there
exists a unique \(\vu^*\) that is orthogonal to the pi subspace, i.e.
\(\mV^T\vu^*=0\). This is useful in the case where we cannot measure \(\pi\)
directly, but must instead observe \(\qoi_{\text{obs}}\); the physical qoi subject
to some added noise \(\epsilon\). In Box's \cite{box1966} formulation, this
\(\epsilon\) represents additional, randomly fluctuating lurking variables. In
this case

\begin{equation}\label{eq:nondim}
  \begin{aligned}
    \qoi_{\text{obs}}(\v\tar) &= \qoi(\v\tar) + \epsilon, \\
    \pi_{\text{obs}}(\v\tar) &= \pi(\v\tar) + \epsilon\exp(-\vu^T\v\tar).
  \end{aligned}
\end{equation}

\noindent In principle, the orthogonality condition \(\mV^T\vu=0\) could aid in
separating signal from noise. However, we shall see that so long as the noise
term is unbiased, we may use the heteroskedastic form of \eqref{eq:nondim}.

\subsection{Pinned variables}
\label{sec:org9267596}
Above, we have implicitly assumed that the \(\v\tar_{\EX}\) are varied
experimentally. In some cases, a variable is known but intentionally not varied
by the experimentalist. We call these known but fixed quantities \emph{pinned
variables}. The decision to pin a variable may be due to cost or safety
constraints. With some modifications, Dimensional Analysis may still be employed
for lurking variable detection in this setting. For a pinned variable, we have
knowledge of its physical dimensions \(\mD_{\PI}\). We may split \(\v\tar\) in a
form similar to \eqref{eq:split}

\begin{equation}
  \v\tar^T = (\v\tar_{\EX}^T,\v\tar_{\LU}^T,\v\tar_{\PI}^T),
\end{equation}

\noindent note that \(\mD=[\mD_{\EX},\mD_{\LU},\mD_{\PI}]\), and write a relation
analogous to \eqref{eq:dim-match}

\begin{equation}\label{eq:pinned-eq}
  \mD_{\EX}\nabla_{\v\tar_{\EX}}\pi = -\mD_{\PI}\nabla_{\v\tar_{\PI}}\pi %
                                      -\mD_{\LU}\nabla_{\v\tar_{\LU}}\pi.
\end{equation}

\noindent If no lurking variables exist, then the quantity
\(\mD_{\EX}\nabla_{\v\tar_{\EX}}\pi\) is expected to lie in the range of
\(\mD_{\PI}\). One can use this information to construct a detection procedure, so
long as \(\cR(\mD_{\LU})\not\subseteq\cR(\mD_{\PI})\). In the case where
\(\cR(\mD_{\PI})=\R{\ndim}\), detecting lurking variables via dimensional analysis
using \eqref{eq:pinned-eq} is impossible, as
\(\cR(\mD_{\LU})\subseteq\cR(\mD_{\PI})\).

Suppose \(\cR(\mD_{\PI})\subset\R{\ndim}\) with
\(\text{Rank}(\mD_{\PI})=r_{\PI}<\ndim\), and let
\(\mW_{\PI}\in\R{\ndim\times(\ndim-r_{\PI})}\) be an orthonormal basis for the
orthogonal complement of \(\mD_{\PI}\); that is
\(\mW_{\PI}^T\mW_{\PI}=\mI_{(\ndim-\rnk_{\PI})\times(\ndim-\rnk_{\PI})}\) and
\(\mW^T_{\PI}\mD_{\PI}=0\). Then we have

\begin{equation}\label{eq:pinned-proj-eq}
  \mW_{\PI}^T\mD_{\EX}\nabla_{\v\tar_{\EX}}\pi = -\mW_{\PI}^T\mD_{\LU}\nabla_{\v\tar_{\LU}}\pi.
\end{equation}

\noindent So long as \(\mW_{\PI}^T\mD_{\LU}\neq0\), \Cref{eq:pinned-proj-eq} may
enable lurking variable detection. Note that if
\(\cR(\mD_{\PI})\notperp\cR(\mD_{\LU})\), then multiplying by \(\mW_{\PI}^T\) will
eliminate some information in \(\mD_{\LU}\nabla_{\v\tar_{\LU}}\pi\). Thus, while
detection is possible with \eqref{eq:pinned-proj-eq}, interpreting the
dimension vector requires more care.

\section{Constructing a Detection Procedure}
\label{sec:orgb99b5a2}
  This section details the requisite machinery for our experimental lurking
variable detection procedures, which ultimately consist of an experimental
design coupled with a hypothesis test. Stein's lemma guides the design and
enables the definition of our null and alternative hypotheses, while Hotelling's
\(T^2\) test provides a suitable statistic.

\subsection{Stein's Lemma}
\label{sec:org501de31}
   Computing the gradient is challenging in the context of a physical
experiment. The usual approximation techniques of finite differences can be
inappropriate in this setting, where experimental noise may dwarf perturbations
to the factors, and arbitrary setting of levels may be impossible. We do not
address the latter issue and assume continuous variables, but attack the former
by pursuing a different tack, that of approximating the average gradient
through experimental design.

Stein's lemma was originally derived in the context of mean estimation; in our
case, we will use it to derive an estimate of the average gradient from point
evaluations of the qoi.\cite{stein1981estimation,lehmann2006theory} If \(X\sim
N(\v\mu,\m\Sigma)\) where \(N(\v\mu,\m\Sigma)\) is a multivariate normal
distribution with mean \(\v\mu\) and invertible covariance matrix \(\m\Sigma\),
Stein's lemma states

\begin{equation}
  \E[\nabla_{\v\tar}f(X)] = \m\Sigma^{-1}\E[(X-\v\mu)f(X)].
\end{equation}

\noindent We will assume in what follows that the exposed parameters are
independent and free to be drawn according to \(\v\tar_{\EX}\sim
N(\v\mu_{\EX},\m\Sigma_{\EX})\), with \(\v\mu_{\EX}\) and \(\m\Sigma_{\EX}\)
invertible selected by the experimentalist, in order to study a desired range of
values. Furthermore, we assume our quantity of interest \(\qoi\) is dimensional,
and subject to \(\v\tar_{\EX}\) -independent, zero-mean noise \(\epsilon\). Taking
an expectation with respect to both \(\v\tar_{\EX}\) and \(\epsilon\), and applying
Stein's lemma to the second line of \eqref{eq:nondim} yields

\begin{equation}
  \begin{aligned}
  \E[\m\Sigma_{\EX}^{-1}(\v\tar_{\EX}-\v\mu_{\EX})\pi_{\text{obs}}(\v\tar_{\EX},\v\tar_{\LU})] &= %
  \E[\nabla_{\v\tar_{\EX}}\pi(\v\tar_{\EX},\v\tar_{\LU})] \\
  &+ \E[\nabla_{\v\tar_{\EX}}\exp(-\vw_{\EX}^T\v\tar_{\EX})\epsilon].
  \end{aligned}
\end{equation}

\noindent Note that the noise term vanishes due to its zero-mean.
Multiplying by the dimension matrix yields

\begin{equation}\label{eq:nondim_stein}
  \E[\mD_{\EX}\m\Sigma_{\EX}^{-1}(\v\tar_{\EX}-\v\mu_{\EX})\pi_{\text{obs}}(\v\tar_{\EX},\v\tar_{\LU})] = %
  \E[\mD_{\EX}\nabla_{\v\tar_{\EX}}\pi(\v\tar_{\EX},\v\tar_{\LU})].
\end{equation}

\noindent \Cref{eq:nondim_stein} enables the definition of an
appropriate null hypothesis \(H_0\) for lurking variable testing. We denote by
\(H_0\) the null hypothesis of no dimensional lurking variables, and by \(H_1\) the
alternative of present lurking variables. These hypotheses are defined by

\begin{equation}\label{eq:hyp}
  \begin{aligned}
    H_0: \E[\mD_{\EX}\m\Sigma_{\EX}^{-1}(\v\tar_{\EX}-\v\mu_{\EX})\pi_{\text{obs}}(\v\tar_{\EX},\v\tar_{\LU})] = 0, \\
    H_1: \E[\mD_{\EX}\m\Sigma_{\EX}^{-1}(\v\tar_{\EX}-\v\mu_{\EX})\pi_{\text{obs}}(\v\tar_{\EX},\v\tar_{\LU})] = \v\nu. \\
  \end{aligned}
\end{equation}

\noindent  Our procedure for dimensional lurking
variables is built upon a test for zero mean of a multivariate distribution. To
define a test statistic, we turn to Hotelling's \(T^2\) test.

\subsection{Hotelling's T-squared test}
\label{sec:orgb65b5b0}
   Hotelling's \(T^2\) test is a classical multivariate generalization of the
t-test.\cite{anderson1958introduction} We draw \(n>d\) samples according to the
design \(\v\tar_{\EX,i}\sim N(\v\mu_{\EX},\m\Sigma_{\EX})\) with \(i=1,\dots,n\),
and based on evaluations

\begin{equation}\label{eq:evaluations}
\vg_i = \mD_{\EX}\m\Sigma_{\EX}^{-1}(\v\tar_{\EX,i}-\v\mu_{\EX})\pi_{\text{obs}}(\v\tar_{\EX,i},\v\tar_{\LU})\in\R{\ndim},
\end{equation}

\noindent define our \(t\) statistic via

\begin{equation}\label{eq:t-stat}
  t^2 = n\bar{\vg}^T\hat{\mS}^{-1}\bar{\vg},
\end{equation}

\noindent where \(\bar{\vg}\) is the sample mean of the \(\vg_i\), and \(\hat{\mS}\)
is the sample covariance

\begin{equation}
  \hat{\mS} = \frac{1}{n-1}\sum_{i=1}^n(\vg_i-\bar{\vg})(\vg_i-\bar{\vg})^T.
\end{equation}

\noindent Under a normal \(\vg_i\) assumption and \(H_0\), \Cref{eq:t-stat}
follows the distribution

\begin{equation}\label{eq:t-dist-null}
  t^2\sim T^2_{\ndim,n-1} = \frac{\ndim(n-1)}{n-\ndim}F_{\ndim,n-\ndim},
\end{equation}

\noindent where \(F_{\ndim,n-\ndim}\) is the central F-distribution with
parameters \(\ndim\) and \(n-\ndim\). Note that even with \(\v\tar_{\EX}\) and
\(\epsilon\) normal, the \(\vg_i\) are not necessarily normal, as \(\qoi(\v\tar)\) may
be an arbitrary function. While the derivation of the reference distribution
\eqref{eq:t-dist-null} requires normality, it is well-known that the
associated test is robust to departures from this
assumption.\cite{kariya1981,mardia1975} Given this robustness,
\eqref{eq:t-dist-null} allows effective testing of \(H_0\). This claim will be
further substantiated in Section \ref{sec:orgbc56bcc} below.

Suppose we are testing at the \(\alpha\) level; let \(F_{\ndim,n-\ndim}^c(\alpha)\)
be the critical test value based on the inverse CDF for \(F_{\ndim,n-\ndim}\). We
reject \(H_0\) if

\begin{equation}
  t^2 \geq \frac{\ndim(n-1)}{n-\ndim}F_{\ndim,n-\ndim}^c(\alpha).
\end{equation}

\subsection{Factors affecting power}
\label{sec:orgd0a7274}
From \eqref{eq:dim-match} and \eqref{eq:nondim_stein}, we know \(\v\nu =
-\E[\mD_{\LU}\nabla_{\v\tar_{\LU}}\pi]\). Note that the first two moments of
\(\vg\) are

\begin{equation}\begin{aligned}\label{eq:g-moments}
  \E[\vg] &\equiv \v\nu = \E[\mD_{\EX}\nabla_{\v\tar_{\EX}}\pi], \\
  \V[\vg] &= \E[\vg\vg^T] - \v\nu\v\nu^T.
\end{aligned}\end{equation}

\noindent In the case where \(H_1\) holds, under a normal \(\vg\) we find that the
test statistic \eqref{eq:t-stat} instead
follows\cite{anderson1958introduction}

\begin{equation}\label{eq:t-dist-alt}
  t^2\sim \frac{\ndim(n-1)}{n-\ndim}F_{\ndim,n-\ndim}(\Delta),
\end{equation}

\noindent where \(F_{\ndim,n-\ndim}(\Delta)\) is a non-central F-distribution with
non-centrality parameter

\begin{equation}\label{eq:non-central}
  \Delta = n\E[\vg]^T\V[\vg]^{-1}\E[\vg],
\end{equation}

\noindent which implies \(\Delta\geq0\). The power \(P\) of our test is then given
by

\begin{equation}\label{eq:power}
  P \equiv \P\left[F_{\ndim,n-\ndim}(\Delta)\geq F_{\ndim,n-\ndim}^c(\alpha) \left| H_1\right.\right].
\end{equation}

\noindent \Cref{eq:power} implies that larger values of \(\Delta\) lead to
higher power. To better understand the contributions to \eqref{eq:non-central},
we apply the Sherman-Morrison formula to \(\V[\vg]\) to separate contributions of
\(\v\nu\) and other variance components. Doing so yields

\begin{equation}\label{eq:non-central-k}
  \Delta = n\left(k + \frac{k^2}{1-k}\right),
\end{equation}

\noindent where \(k = \v\nu^T\E[\vg\vg^T]^{-1}\v\nu\). It is easy to see that
\(k\in[0,1]\). Based on \eqref{eq:non-central-k}, we see that power is an
increasing function of \(n\) and \(k\). However, Taylor expanding \(\pi(\v\tar)\)
about \(\v\mu\) reveals that \(\v\nu\) and \(\E[\vg\vg^T]\) have different
dependencies on the moments of the sampling distribution, and therefore
different dependencies on \(\v\mu_{\EX},\m\Sigma_{\EX}\). Thus, the dependence of
power (via \(k\)) on the sampling distribution cannot be known without more
knowledge of the functional form of \(\pi(\v\tar)\). For instance, it is easy to
construct simple examples of \(\pi(\v\tar)\) which limit to either power of zero
or power of one while scaling \(\m\Sigma_{\EX}\) towards zero or infinity.

To see the impact of noise on power, we consider the average outer-product of
the test statistic, given by

\begin{equation}
  \E[\vg\vg^T] = \mD_{\EX}\m\Sigma_{\EX}^{-1}\E[(\v\tar_{\EX}-\v\mu_{\EX})(\v\tar_{\EX}-\v\mu_{\EX})^T(\pi^2+\exp(-2\vw_{\EX}^T\v\tar_{\EX})\tau^2)]\m\Sigma_{\EX}^{-1}\mD_{\EX}^T,
\end{equation}

\noindent where \(\tau^2\) is the variance of \(\epsilon\). Through multiple
applications of Jensen's inequality, we may show

\begin{equation}\label{eq:noise-term}
  \mD_{\EX}\m\Sigma_{\EX}^{-1}\E[(\v\tar_{\EX}-\v\mu_{\EX})(\v\tar_{\EX}-\v\mu_{\EX})^T\exp(-2\vw^T_{\EX}\v\tar_{\EX})\tau^2]\m\Sigma_{\EX}^{-1}\mD_{\EX}^T \geq (\mD_{\EX}\vw_{\EX})(\mD_{\EX}\vw_{\EX})^T\tau^2/\bar{Q}^2 \geq 0,
\end{equation}

\noindent where \(\bar{Q}\equiv\exp(-\vw_{\EX}^T\v\mu)\). Since \(\E[\vg\vg^T]\)
enters as an inverse, \Cref{eq:noise-term} shows that the noise variability
\(\tau\) negatively affects power, as is intuitively expected. However, since this
term is added with the qoi-dependent term, one cannot know an appropriate scale
for \(\tau\) without more knowledge of the structure of \(\pi(\v\tar)\). We will see
below that the effects of noise can be similar across disparate values of
\(\tau/\bar{Q}\).

\section{Detecting Lurking Variables}
\label{sec:org26bf4ba}
  In this section, we present procedures for lurking variable detection. The
first procedure is based entirely upon a priori information. The second
combines Dimensional Analysis (Sec. \ref{sec:org6f9d706}) with the statistical
machinery introduced above (Sec. \ref{sec:orgb99b5a2}). The final
procedure is a modification of the second, introduced to handle pinned
variables. Examples of these procedures are presented below, in Section \ref{sec:orgbc56bcc}.

\subsection{Detection with a priori information}
\label{sec:org01d2a4a}
   In some cases, the analyst may detect lurking variables based solely on a
priori information. This can be done with a simple analytic check for
dimensional homogeneity. In the case where no non-dimensionalizing factor can be
defined; that is \(\vd(\qoi)\not\in\cR(\mD_{\EX})\), dimensional homogeneity
cannot hold. This is a clear signal that lurking variables affect our qoi.
\subsection{Detection with a physical qoi}
\label{sec:org93541d4}
   This subsection lays out a three-step procedure to test for lurking
variables. Note that in what follows, the lurking variables need not vary, and
are assumed to be fixed throughout the experiment.

This procedure assumes the following setting: Let \(\qoi\) be a physical quantity
of interest, with identified predictors \(\v\tar_{\EX}\in\R{\nvar_{\EX}}\) and
physical dimensions \(\mD_{\EX}\in\R{\ndim\times \nvar_{\EX}}\). Assume
\(\vd(\qoi)\in\cR(\mD_{\EX})\), and that one may vary the \(\v\tar_{\EX}\) within
acceptable bounds and evaluate \(\qoi_{\text{obs}}(\v\tar_{\EX})\) via the
experimental setup.

\bigskip
\noindent \textbf{Step 1: Perform dimensional analysis}

Solve \(\vd(\qoi)=\mD_{\EX}\vw_{\EX}\) for \(\vw_{\EX}\in\R{\nvar_{\EX}}\); this
enables computation of

$$\pi_{\text{obs}}(\v\tar_{\EX}) = \qoi_{\text{obs}}(\v\tar_{\EX})\exp(-\vw_{\EX}^T\v\tar_{\EX}).$$

\noindent \textbf{Step 2: Design and perform experiment}

Choose \(\v\mu_{\EX}\in\R{\ndim_{\EX}}\) and
\(\m\Sigma_{\EX}\in\R{\nvar_{\EX}\times \nvar_{\EX}}\) in order to select a range
of values for study. Draw samples \(\v\tar_{\EX}\sim
N(\v\mu_{\EX},\m\Sigma_{\EX})\) with \(i=1,\dots,n\). Evaluate
\(\qoi_{\text{obs}}(\v\tar_{\EX})\) via the experimental setup, and use
\(\vw_{\EX}\) from Step 1 to compute \(\pi_{\text{obs}}(\v\tar_{\EX})\).

\noindent \textbf{Step 3: Test}

Select a confidence level, form the \(\vg_i\) defined by \eqref{eq:evaluations},
compute \(t^2\) via \eqref{eq:t-stat}, and compare against the reference
distribution defined by \eqref{eq:t-dist-null}. If \(t^2\) is larger than
the critical value, reject the null hypothesis of no lurking variables.
\subsection{Detection in the presence of pinned variables}
\label{sec:org0982799}
   If a pinned variable affects our qoi, the detection procedure defined above
is inappropriate. In this case, we recommend a simple modification to the
procedure above, based on \eqref{eq:pinned-proj-eq}. In what follows, we
assume the same setting as Section \ref{sec:org93541d4}, with the
addition of \(\v\tar_{\PI}\in\R{\nvar_{\PI}}\) pinned variables, with known
dimensions \(\mD_{\PI}\in\R{\ndim\times \nvar_{\PI}}\). These \(\v\tar_{\PI}\)
remain fixed throughout the experiment. As mentioned in Section \ref{sec:org9267596}, interpretation of the dimension vector requires more care.

\bigskip
\noindent \textbf{Step 1: Perform dimensional analysis}

Determine \(r_{\PI}=\text{Rank}(\mD_{\PI})\), and check that \(r_{\PI}<\ndim\). If
so, compute a basis \(\mW_{\PI}\in\R{\ndim\times(\ndim-\rnk_{\PI})}\) for the
orthogonal complement of \(\mD_{\PI}\), e.g. via a QR decomposition.

\noindent \textbf{Step 2: Design and perform experiment}

This step remains unchanged; note that \(\v\tar_{\EX}\) refers only to those
variables which can be varied.

\noindent \textbf{Step 3: Test}

Follow Step 3 above, but modify the \(\vg_i\) as defined by \eqref{eq:evaluations}

\begin{equation}\label{eq:evaluations-proj}
  \vg'_i = \mW_{\PI}^T\vg_i.
\end{equation}

\noindent Compute \(t^2\) via \eqref{eq:t-stat} based on the \(\vg'_i\); note
that \(\ndim\) is replaced with \(\ndim-\rnk_{\PI}\) for this modified statistic. If
\(t^2\geq
\frac{(\ndim-r_{\PI})(n-1)}{n-\ndim+r_{\PI}}F_{\nvar_{\PI},n-\ndim+r_{\PI}}^c(\alpha)\),
then reject \(H_0\).
\section{Detection Examples}
\label{sec:orgbc56bcc}
\subsection{Physical examples}
\label{sec:org9de8c1b}
   In what follows, we consider two physical examples motivated by engineering
applications: Rough Pipe Flow and Two-Fluid Flow. This section gives a short
description of the examples; full details, including derivations and code
sufficient to reproduce all results below, are available in the Supporting
Material. Rough Pipe Flow has already been introduced above in Section
\ref{sec:org70885ee}; the full list of the predictors, the response, and
associated physical dimensions is detailed in Table \ref{tab:rough-vars}. For
Rough Pipe Flow, evaluation of the qoi is based on analytic and empirical
relationships: Poiseuille's law for laminar flow, and the Colebrook equation to
model behavior from the turbulent onset to full
turbulence.\cite{white2011,colebrook1937}

\begin{table}[!ht]
\centering
  \begin{tabular}{@{}lcc@{}}
  \toprule
  Physical Variable       & Symbol               & Physical Dimensions\\
  \midrule
  Pressure Gradient (qoi) & $\frac{\Delta P}{L}$ & $M^1L^{-1}T^{-2}$ \\
  Pipe Diameter           & $d_P$                & $L$ \\
  Pipe Roughness          & $\epsilon_P$         & $L$ \\
  Fluid Bulk Velocity     & $U_F$                & $L^1T^{-1}$ \\
  Fluid Density           & $\rho_F$             & $M^1L^{-3}$ \\
  Fluid Viscosity         & $\mu_F$              & $M^1L^{-1}T^{-1}$ \\
  \bottomrule
  \end{tabular}
\caption{Physical variables for Rough Pipe Flow.}
\label{tab:rough-vars}
\end{table}

The second example is Two-Fluid Flow, inspired by an engineering need to pump
viscous fluids at a high rate. In this setting, two immiscible fluids are
assumed to be in steady laminar flow through a channel. The fluids form layers
depicted in Figure \ref{fig:two-fluid-schematic}, where the inner fluid is
assumed to be more viscous than the outer fluid (\(\mu_i>>\mu_o\)). This outer
lubricating layer allows faster transport of the inner fluid, but reduces the
effective diameter for pumping. The qoi is the volumetric flow rate of the inner
flow. In this example, evaluation of the qoi is based on an analytic expression,
derived from the Navier-Stokes equations with some standard simplifying
assumptions.\cite{white2011}

\begin{figure}[!ht]
\centering
\includegraphics[width=0.75\textwidth]{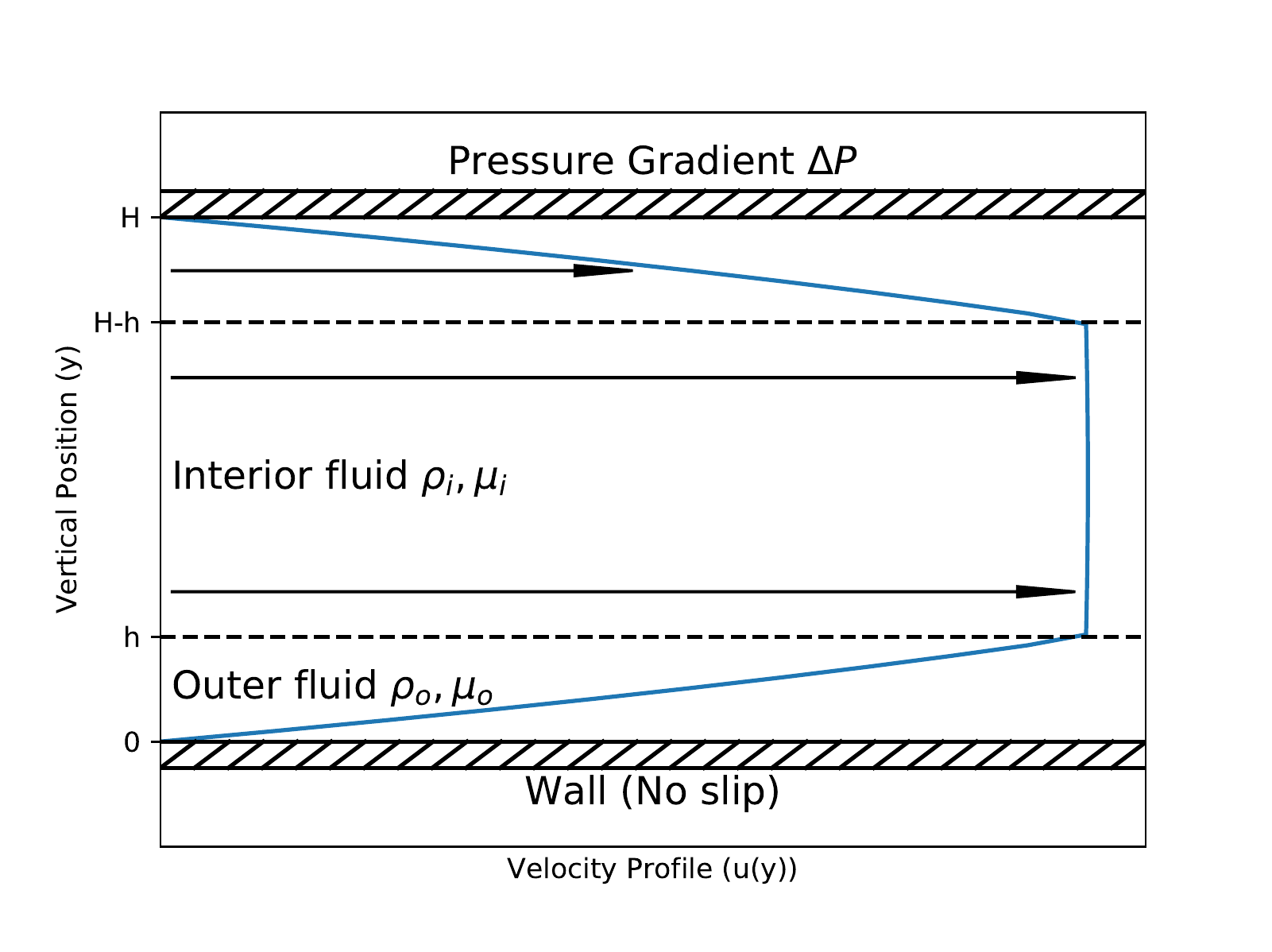}
\caption{Schematic for Two-Fluid Flow. The view is of a cross-section
  of an infinite channel formed by two parallel plates. A viscous fluid flows
  in the cavity between these two surfaces. Fluid flows from left to right,
  and the velocity profile (depicted in blue) across the channel is shown;
  zero velocity corresponds to the left boundary of the figure, while positions
  on the curve further to the right correspond to greater velocity. The dashed
  horizontal lines illustrate the boundary between the inner and outer fluids;
  the qoi is the flow rate between the two dashed lines. By design, the flow
  rate is nearly independent of the inner flow fluid properties.}
\label{fig:two-fluid-schematic}
\end{figure}

In this example, the qoi is significantly less sensitive to the inner fluid
properties than the other factors. This is by design. Since our lurking variable
detection procedure is based on the sensitivity of the qoi to the lurking
variable in question, detecting a lurking inner fluid property is challenging.
Note also that formally, the dimensional qoi is not sensitive to the fluid
densities. These predictors are included to demonstrate that our procedure
handles such unimportant variables automatically.

\begin{table}[!ht]
\centering
  \begin{tabular}{@{}lcc@{}}
  \toprule
  Physical Variable         & Symbol     & Physical Dimensions\\
  \midrule
  Flow Rate (qoi)           & $\qoi$     & $L^2T^{-1}$ \\
  Applied Pressure Gradient & $\nabla P$ & $M^1L^{-2}T^{-2}$ \\
  Outer Fluid Thickness     & $h$        & $L$ \\
  Inner Fluid Thickness     & $H$        & $L$ \\
  Outer Fluid Viscosity     & $\mu_o$    & $M^1L^{-1}T^{-1}$ \\
  Inner Fluid Viscosity     & $\mu_i$    & $M^1L^{-1}T^{-1}$ \\
  Outer Fluid Density       & $\rho_o$   & $M^1L^{-3}$ \\
  Inner Fluid Density       & $\rho_i$   & $M^1L^{-3}$ \\
  \bottomrule
  \end{tabular}
\caption{Physical variables for Two-Fluid Flow}
\label{tab:two-vars}
\end{table}

\subsection{Analytic detection}
\label{sec:org890b95a}
As a concrete example of the analytic detection procedure, consider the problem
of Rough Pipe Flow. Suppose an analyst believes that the pipe diameter \(d_P\) and
fluid bulk velocity \(U_F\) are the only factors which affect the qoi. In this
case, the reduced dimension matrix is given in Table
\ref{tab:pipe-reduced-success}.

\begin{table}[!ht]
  \centering
  \begin{tabular}{@{}lccc}
    \toprule
    Dimension  & $d_P$ & $U_F$ & $\Delta P$ \\
    \midrule
    Mass (M)   &   0      &  0    &  1  \\
    Length (L) &   1      &  1    & \mm1 \\
    Time (T)   &   0      & \mm1  & \mm2 \\
    \bottomrule
  \end{tabular}
  \caption{Reduced dimension matrix for Rough Pipe Flow. The center columns are $\mD_{\EX}$, while
the rightmost column is $\vd(\qoi)$. In this case, it is evident that dimensional homogeneity
cannot hold, and a lurking variable must exist.}
  \label{tab:pipe-reduced-success}
\end{table}

\noindent By inspection, we can see \(\vd(\qoi)\not\in\cR(\mD_{\EX})\). One cannot
form a non-dimensionalizing factor from the given predictors; clearly something
is missing. Note that none of the proposed predictors has physical dimensions of
Mass; this hints that the lurking variables must introduce a Mass.

The analysis above can be performed without experimentation, but it is extremely
limited. Suppose that an analyst instead proposed predictors of fluid density
\(\rho_F\) and bulk velocity \(U_F\). Then the reduced dimension matrix is given in
Table \ref{tab:pipe-reduced-failure}. In this case, dimensional homogeneity
holds, and the analyst can form a non-dimensionalizing factor; the \emph{dynamic
pressure} \(\f12\rho_FU_F^2\).\cite{white2011} To learn more, the analyst must turn
to fluid mechanics; either analytic or experimental.

\begin{table}[!ht]
  \centering
  \begin{tabular}{@{}lccc@{}}
    \toprule
    Dimension  & $\rho_F$ & $U_F$ & $\Delta P$ \\
    \midrule
    Mass (M)   &   1      &  0    &  1 \\
    Length (L) & \mm3     &  1    & -1 \\
    Time (T)   &   0      & \mm1  & -2 \\
    \bottomrule
  \end{tabular}
  \caption{Reduced dimension matrix for Rough Pipe Flow. The center columns are $\mD_{\EX}$, while
the rightmost column is $\vd(\qoi)$. In this case, one can form the \emph{dynamic pressure}
 $\f12\rho_F U_F^2$, which is a suitable non-dimensionalizing factor for the qoi. To learn
more, an experimentalist could collect data to probe the the functional relationship between
the response and predictors. This would reveal additional variability not predicted by
naive Dimensional Analysis.}
  \label{tab:pipe-reduced-failure}
\end{table}

\subsection{Experimental detection}
\label{sec:orgd31fd60}
In this section, we perform numerical experiments to test the assumptions
introduced in Section \ref{sec:org26bf4ba} above, and assess the efficacy
of the proposed detection procedures. In the following examples, we query an R
implementation of the models above to perform virtual
experiments.\cite{r2017,wickham2009,venables2002,sundar2014,csardi2016} To mimic
experimental variability, we add zero-mean Gaussian noise to \(\qoi\) with a
chosen standard deviation \(\tau\); we increase \(\tau\) in each case until a
substantial degradation in power is observed. We simulate lurking variables by
choosing a subset of the variables and fixing them during the experiment. We
present various cases of lurking or pinned variables, in order to demonstrate
both the ordinary and modified detection procedures.

In all cases, we compute a non-dimensionalizing factor from the exposed
variables according to Appendix \ref{sec:org8e1bcfd}. We present
sweeps through samples drawn and noise variability, with \(N=5000\) replications
at each setting to estimate the Type I error and power of the detection
procedures. A significance level of \(\alpha=0.05\) is used for all examples.

We estimate both Type I error and power as binomial parameters. Type I error is
simulated by considering the case when there are no lurking variables. Power is
estimated by fixing and withholding variables from the analysis to simulate
lurking variables. Since we consider cases where the estimates approach the
extremes of the unit interval, the simple normal approximation is inappropriate
for our purposes. Below we construct intervals with coverage probability \(95\%\)
using Wilson's method, with bounds given by

\begin{equation}
\frac{1}{N+z^2}\left[N_r+\f12z^2\pm z\sqrt{\frac{N_rN_f}{N}+\f14z^2}\right],
\end{equation}

\noindent where \(N_r,N_f\) are the number of replications where we (respectively)
reject or fail to reject, and \(z\) is the \((1-0.95)/2\) quantile of the standard
normal.\cite{wilson1927probable}

\subsubsection{Rough Pipe Flow}
\label{sec:org47db551}
The parameters (Tab. \ref{tab:rough-pipe-param}) of the sampling distribution
\(\v\mu_{\EX},\m\Sigma_{\EX}\) are chosen to emphasize turbulent flow. In this
regime, the roughness of the pipe affects the qoi. Figure \ref{fig:rough-pipe}
presents Type I error and power curves with \(95\%\) confidence intervals. These
results demonstrate Type I error near the requested level, and increasing
detection power with increased sample size. As expected, greater noise
variability leads to less power.

\begin{table}[H]
\centering
\begin{tabular}{@{}llllll@{}}
\toprule
          & $\log(\rho_F)$ & $\log(U_F)$ & $\log(d_P)$ & $\log(\mu_F)$ & $\log(\epsilon_P)$ \\
\midrule
 $\v\mu_{\EX}$    & $0.1682$ & $5.7565$ & $0.3965$ & $\mm11.3102$ & $\mm2.0999$ \\
 $\m\Sigma_{\EX}^{1/2}$ & $0.0561$ & $0.3838$ & $0.0448$ & $0.0676$ & $0.0676$ \\
\bottomrule
\end{tabular}
\caption{Sampling parameters for Rough Pipe Flow. The diagonal variance matrix $\m\Sigma_{\EX}$
  is determined by the given standard deviation components. Logarithms are taken
  using base $10$.}
\label{tab:rough-pipe-param}
\end{table}

\begin{figure}[H]
\begin{minipage}{0.49\textwidth}
  \includegraphics[width=0.95\textwidth]{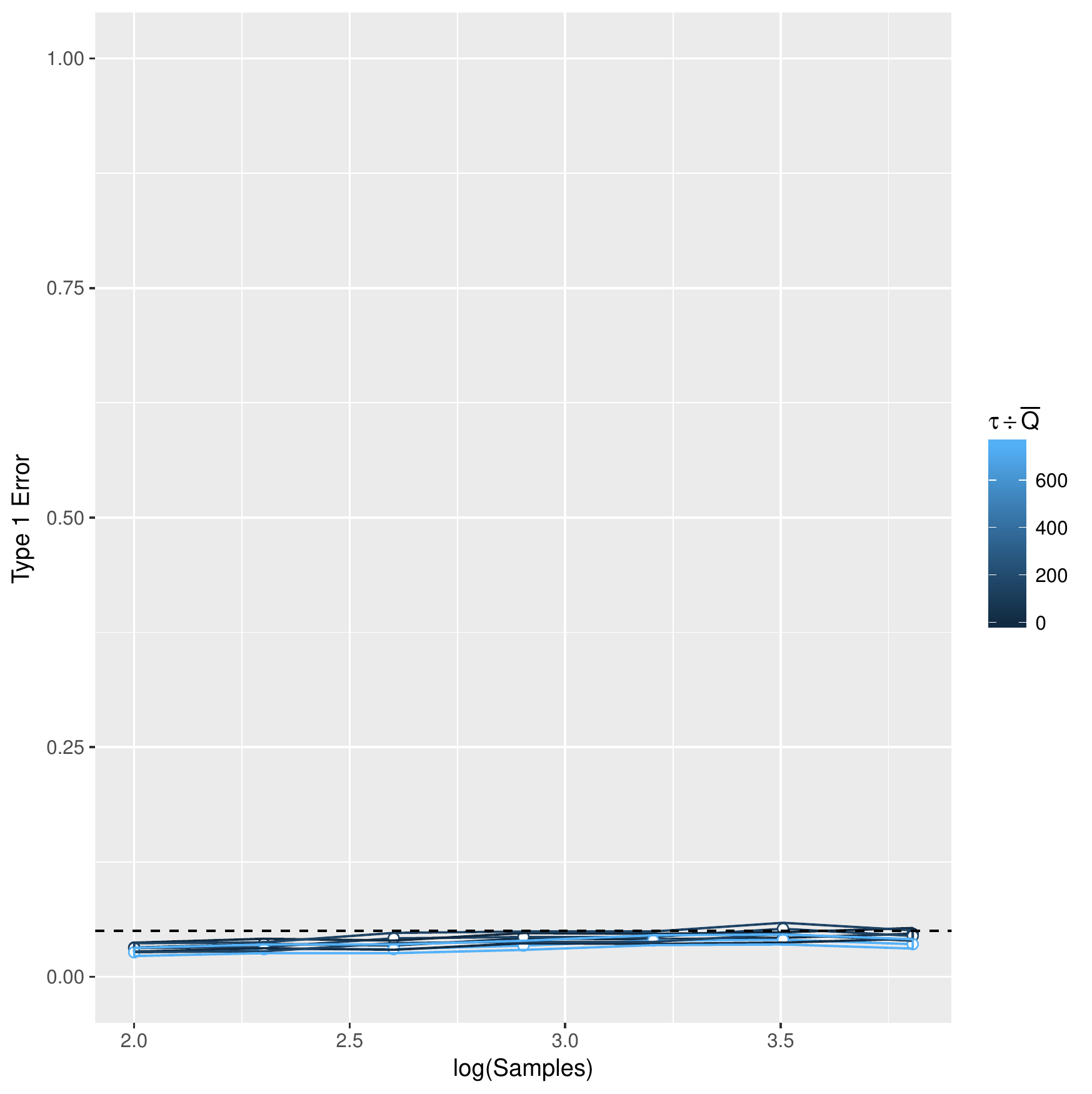}
\end{minipage} %
\begin{minipage}{0.49\textwidth}
  \includegraphics[width=0.95\textwidth]{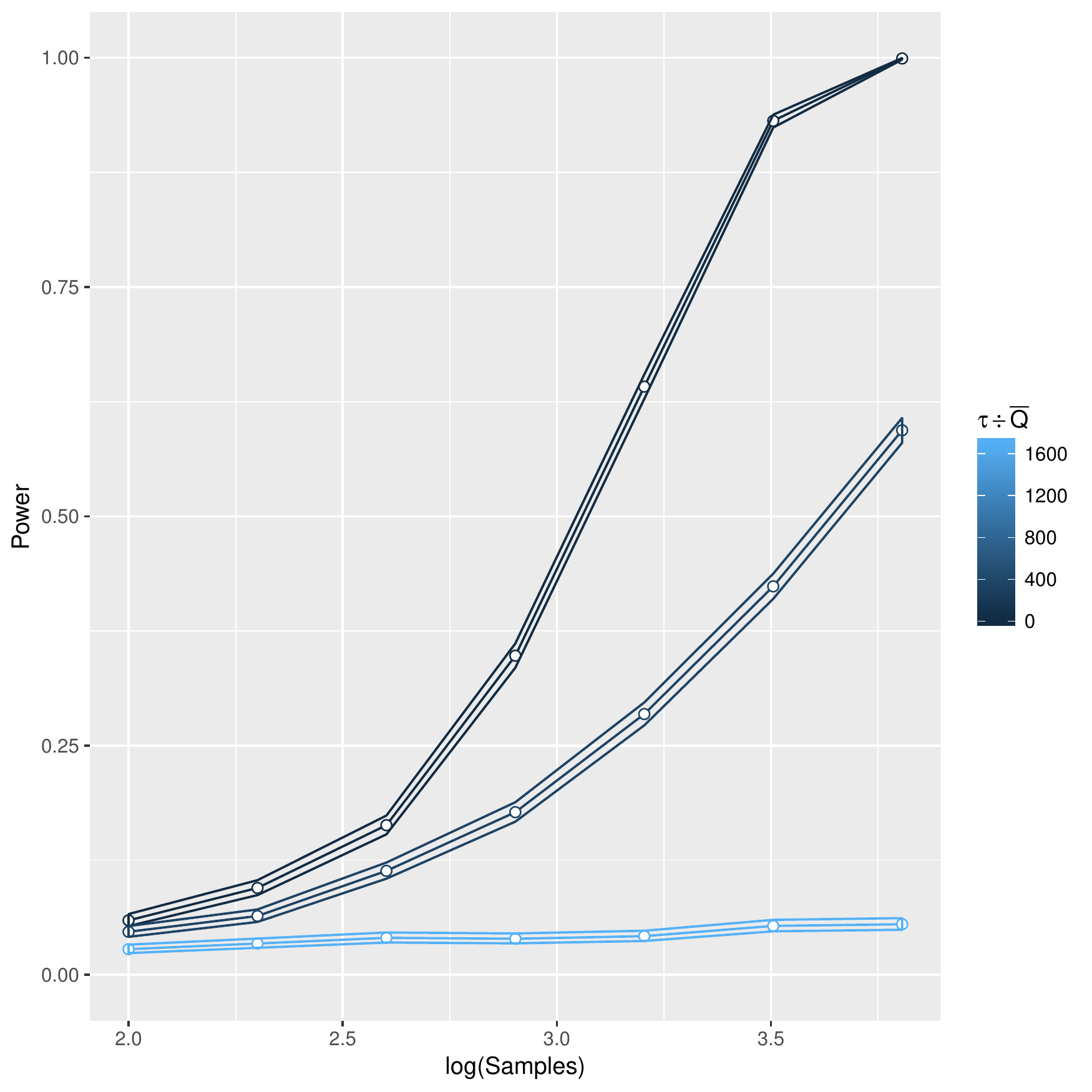}
\end{minipage}
\caption{Rough Pipe Flow Type I error (left) and power (right). The left image
  demonstrates error near the requested level. The right image considers a case
  where the roughness $\epsilon_P$ is a lurking variable. Note that in this
  experiment, the roughness is held fixed at a nominal value, mimicking the
  nature of Reynold's original 1883 experiment. The results shown here suggest
  that an experimentalist could have identified the presence of a lurking variable
  using a statistical procedure informed by Dimensional Analysis, rather than
  employing domain-specific knowledge.}
\label{fig:rough-pipe}
\end{figure}

\newpage
Table \ref{tab:rough-pipe-table} presents moment estimates for the distribution
of p-values in the null-following case, at various settings of \(n\) and \(\tau\).
For an exact reference distribution under the null hypothesis, the p-values
follow the uniform distribution on \([0,1]\), which has mean and variance \(0.5,
1/12\approx0.083\) respectively. The results are compatible with a uniform
distribution of p-values across a wide range of \(n,\tau\), endorsing the
assumptions used to derive the reference distribution. Figure
\ref{fig:org212bb31} depicts the empirical distribution of p-values at
\(n=6400,\tau=100\), enabling a more detailed assessment of our assumptions.

\begin{figure}[htbp]
\centering
\includegraphics[width=0.45\textwidth]{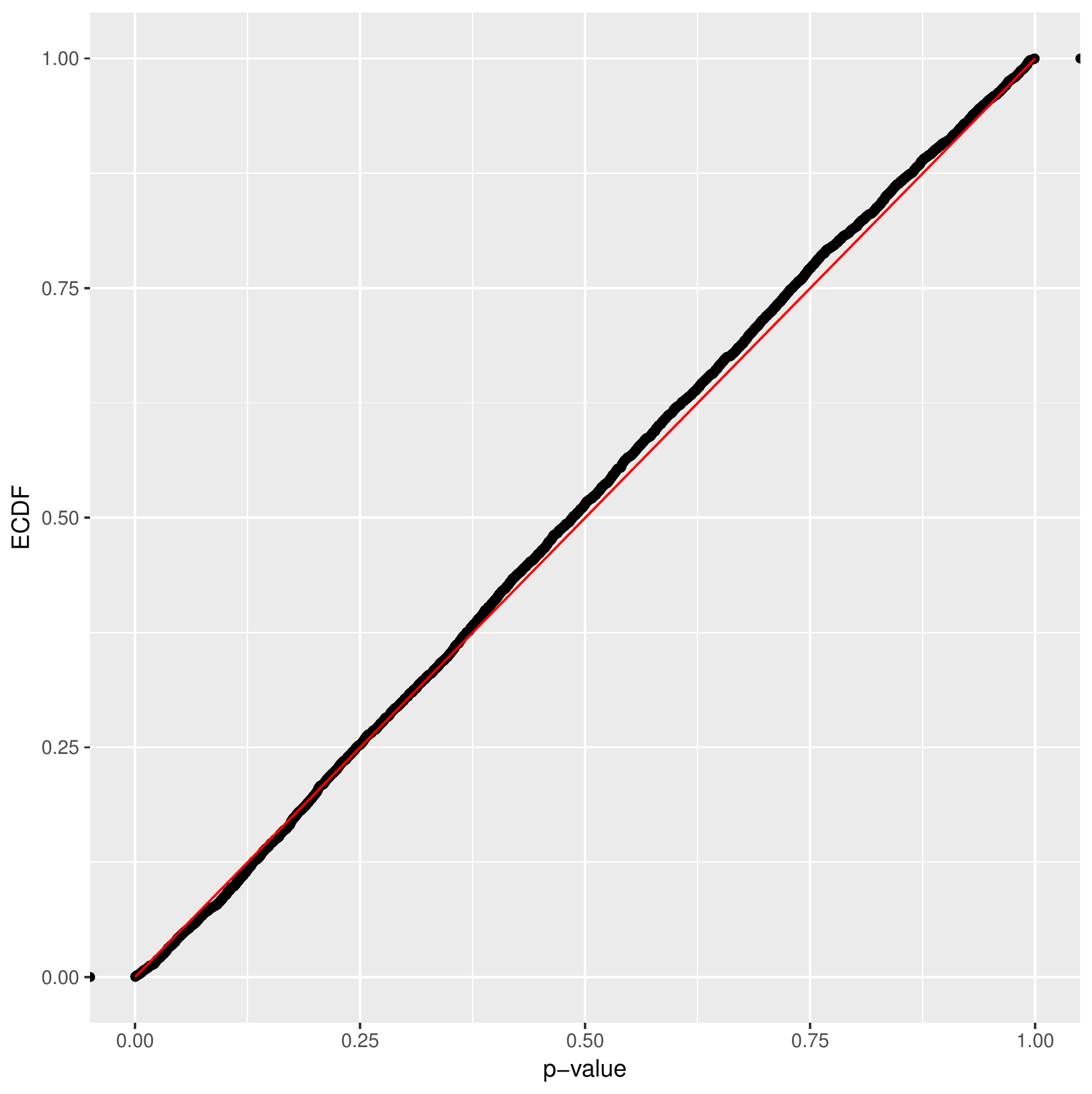}
\caption{\label{fig:org212bb31}
Rough Pipe Flow p-value empirical CDF with no lurking variables, \(n=6400,\tau=100\). Sorted p-values are denoted by dots, with an added red diagonal. The distribution of p-values is approximately uniform, endorsing the choice of reference distribution.}
\end{figure}

\begin{table}[H]
\centering
\input{./data/mom_pipe_raw.r_Ie1,2,3,4,5_Ip_sfac2.00.tex}
\caption{Rough Pipe Flow moment estimates of p-value distribution with no lurking variables.
  Results are presented as pairs of mean, variance. Note that $U(0,1)$ has first moments
  $0.5$ and $1/12\approx0.083$. These results suggest that under these conditions, the
  empirical distribution of p-values is approximately uniform, with the greatest deviations
  occurring at low sample count and high noise variability.}
\label{tab:rough-pipe-table}
\end{table}

\subsubsection{Two-Fluid Flow}
\label{sec:org19150c3}
    The parameters of the sampling distribution \(\v\mu_{\EX},\m\Sigma_{\EX}\) are
chosen to emphasize a large viscosity ratio \(\mu_i>>\mu_o\) and a reasonable
range for the other design parameters. We alternately consider the flow
variables \(\rho,\mu\) as lurking, switching between the inner and outer pairs.
Figures \ref{fig:two-fluid-1} and \ref{fig:two-fluid-2} present Type I error and
power curves. The right image in Figure \ref{fig:two-fluid-1} demonstrates power
indistinguishable from \(\alpha=0.05\) at the studied sample counts \(n\); this is
the case where the inner flow variables are lurking. As noted above, the qoi
only weakly depends on these variables, as per engineering design. This example
demonstrates that some lurking variables are inherently challenging to detect.

\begin{table}[H]
\centering
\begin{tabular}{@{}llllllll@{}}
\toprule
          & $\log(\nabla P)$ & $\log(h)$ & $\log(H)$ & $\log(\mu_o)$ & $\log(\mu_i)$ & $\log(\rho_o)$ & $\log(\rho_i)$ \\
\midrule
 $\v\mu_{\EX}$    & $1.0397$ & $\mm1.7533$ & $0.3466$ & $0.3466$ & $3.8005$ & $0.3466$ & $1.4979$ \\
 $\m\Sigma_{\EX}^{1/2}$ & $0.3466$ & $0.1831$ & $0.1155$ & $0.1155$ & $0.0372$ & $0.1155$ & $0.0372$ \\
\bottomrule
\end{tabular}
\caption{Sampling parameters for Two Fluid Flow. The diagonal variance matrix $\m\Sigma_{\EX}$
  is determined by the given standard deviation components.}
\label{tab:two-fluid-param}
\end{table}

\begin{figure}[!ht]
\begin{minipage}{0.49\textwidth}
  \includegraphics[width=0.95\textwidth]{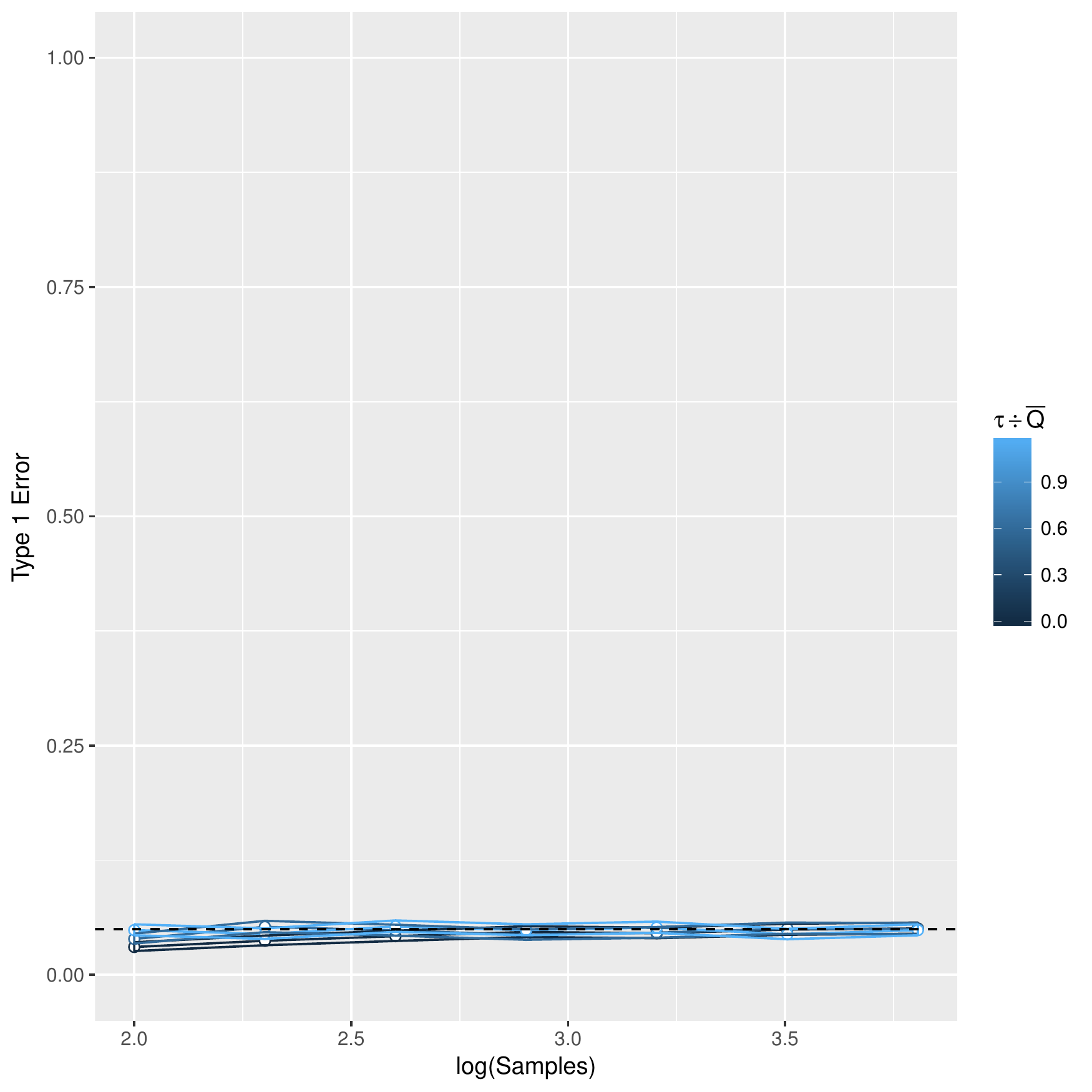}
\end{minipage} %
\begin{minipage}{0.49\textwidth}
  \includegraphics[width=0.95\textwidth]{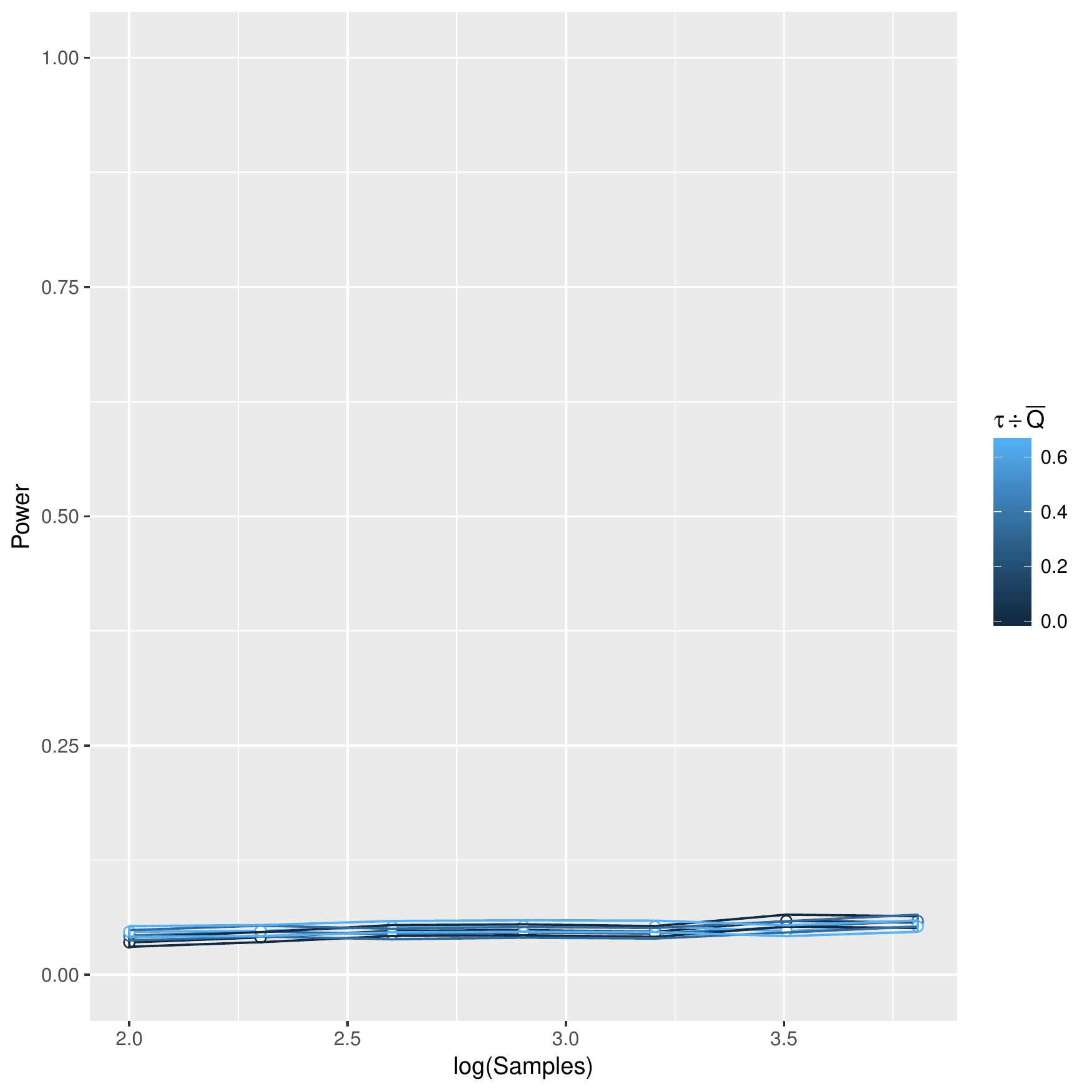}
\end{minipage}
\caption{Two-Fluid Flow Type I error (left) and power (right). The right image considers
  a case where the inner viscosity is lurking. This test case demonstrates that
  some lurking variable are inherently challenging to detect.}
\label{fig:two-fluid-1}
\end{figure}

Figure \ref{fig:two-fluid-2} considers cases where the outer viscosity is
lurking. The right image considers the inner thickness as a pinned variable,
while the left image varies all the exposed variables. Note that the right image
necessitates the modified procedure to address the pinned variable. Figure
\ref{fig:two-fluid-2} demonstrates that for Two-Fluid Flow and this particular
combination of variables, the presence of a pinned variable does not result in a
significant power loss. Table \ref{tab:two-fluid-moments} presents moment
estimates for the distribution of p-values in the case with no lurking
variables.

\begin{figure}[!ht]
\begin{minipage}{0.49\textwidth}
  \includegraphics[width=0.95\textwidth]{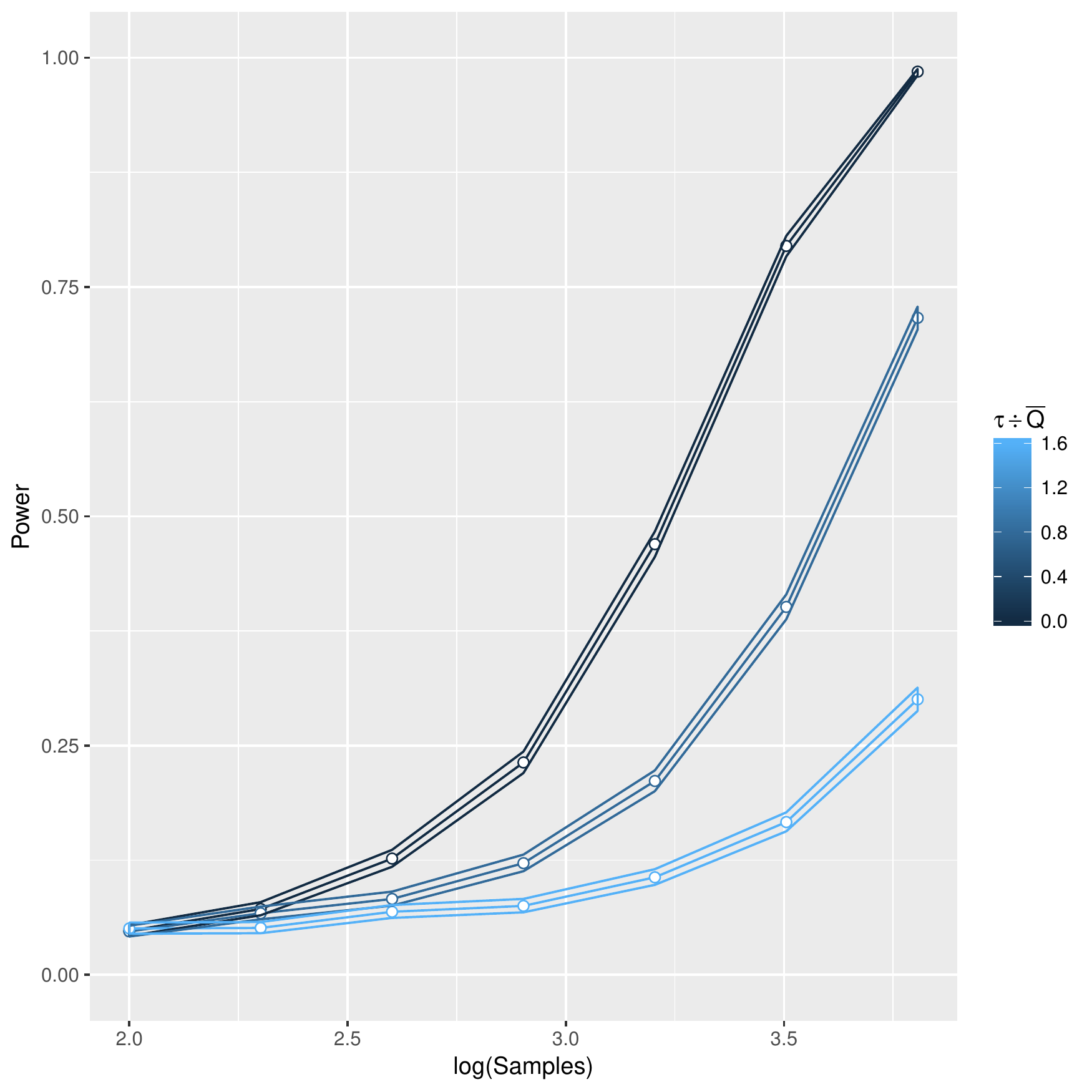}
\end{minipage} %
\begin{minipage}{0.49\textwidth}
  \includegraphics[width=0.95\textwidth]{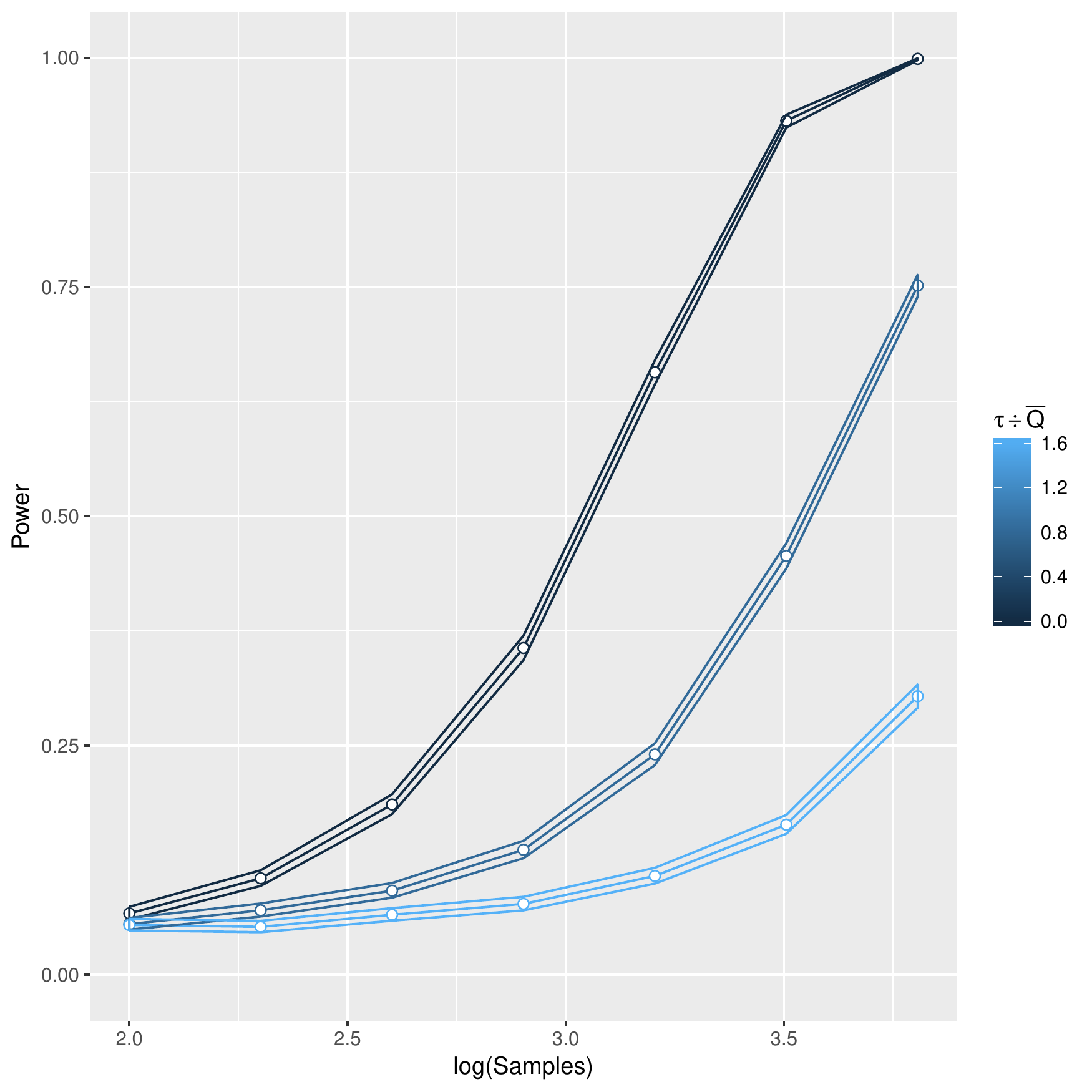}
\end{minipage}
\caption{Two-Fluid Flow power without (left) and with (right) pinned variables.
  Both cases consider the outer viscosity as a lurking variable, while the right
  additionally considers the inner thickness as a pinned variable. In the cases
  considered here, the presence of a pinned variable does not result in a significant
  power loss.}
\label{fig:two-fluid-2}
\end{figure}

\begin{table}[!ht]
\centering
\input{./data/mom_two_fluid_raw.r_Ie1,2,3,4,5,6,7_Ip_sfac6.00.tex}
\caption{Two-Fluid Flow moment estimates of p-value distribution with no
  lurking variables. These results suggest that under these conditions, the
  distribution of p-values is uniform.}
\label{tab:two-fluid-moments}
\end{table}

\begin{figure}[htbp]
\centering
\includegraphics[width=0.45\textwidth]{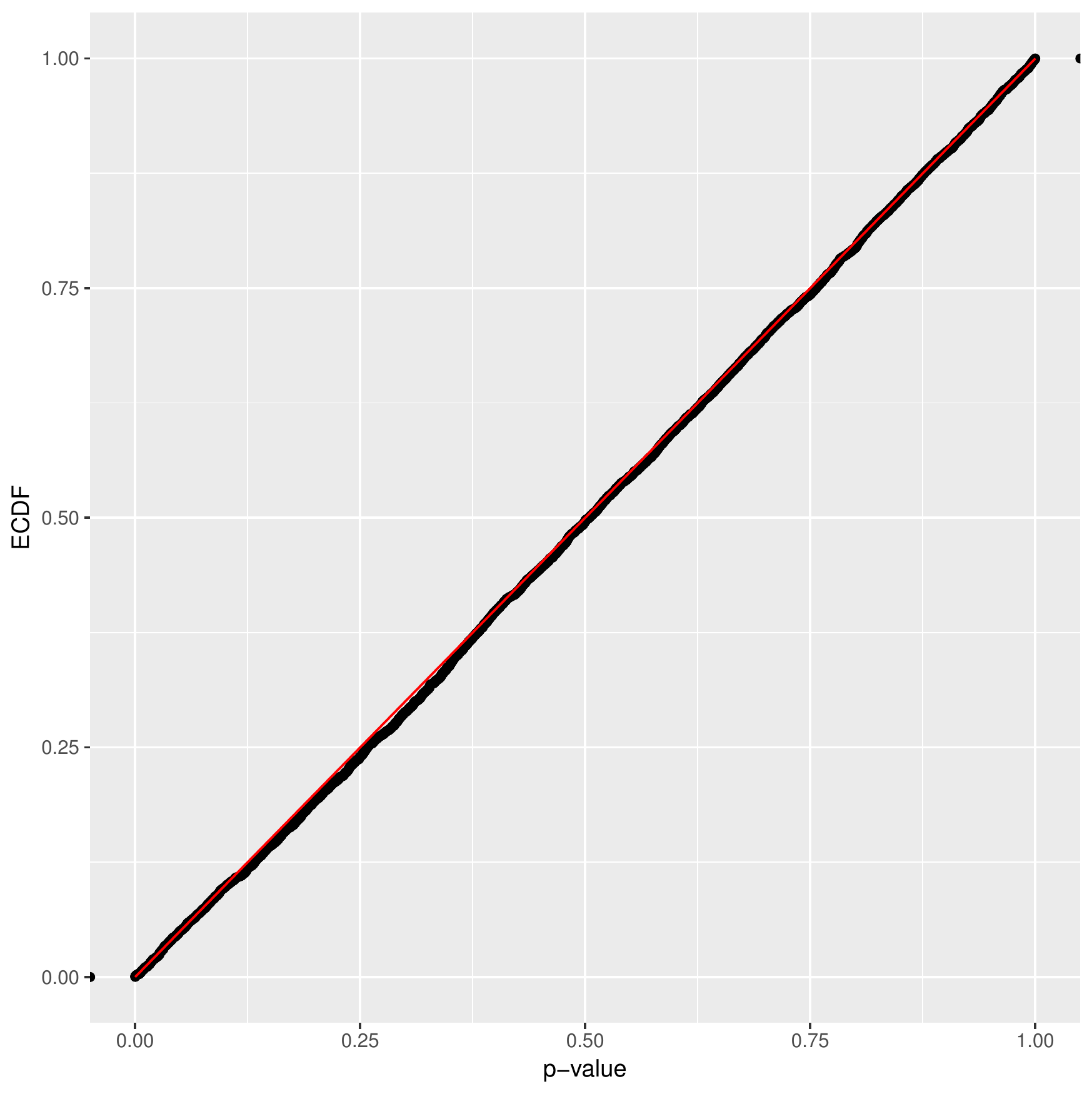}
\caption{\label{fig:org4ce23f5}
Two-Fluid Flow p-value empirical CDF with no lurking variables, \(n=6400,\tau=0.5\). Sorted p-values are denoted by dots, with an added red diagonal. The distribution of p-values is approximately uniform, endorsing the choice of reference distribution.}
\end{figure}

As described above, the quantity \(\v\nu\) contains useful information if
lurking variables exist. To illustrate, we consider realizations of the sample
estimate \(\hat{\v\nu}\) in the two cases considered in Figure
\ref{fig:two-fluid-2}. For comparison, we compute the quantity
\(\mW_{\PI}\mW_{\PI}^T\mD_{\LU}\) and scale this vector to have the same length as
the estimate \(\hat{\v\nu}\); this quantity is reported as \(\vw\). Since there is
only one lurking variable in these cases, we expect that \(\E[\hat{\v\nu}]=\vw\).

\begin{table}[!ht]
\centering
\begin{tabular}{@{}lllll@{}}
\toprule
Case & & M & L & T\\
\midrule
Without pinned var & \(\hat{\v\nu}_L\) & 0.8165 & -0.9752 & -0.7370\\
 & \(\vw_L\) & 0.8487 & -0.8487 & -0.8487\\
\midrule
With pinned var & \(\hat{\v\nu}_R\) & 0.6731 & 0.0000 & -0.6064\\
 & \(\vw_R\) & 0.6406 & 0.0000 & -0.6406\\
\bottomrule
\end{tabular}
\caption{Two-Fluid Flow dimension vector estimates. Note that interpreting the
  dimension vector is relatively straightforward in the case without pinned
  variables; since the lurking variable is a viscosity, the dimension vector
  matches the expected $(1,-1,-1)^T$, up to a scaling constant. While these
  results are encouraging, further work is necessary to provide a formal
  test procedure.}
\label{tab:two-fluid-nu}
\end{table}

\section{Discussion}
\label{sec:org0b49157}
In this article, we presented a modified form of the Buckingham \(\pi\) theorem
suitable for testing the presence of lurking variables. We then constructed
experimental detection procedures based on a sampling plan informed by Stein's
lemma, a reference distribution arising from Hotelling's \(T^2\) test, and
reasonable assumptions on the distribution of the response. We supported these
assumptions through example problems inspired by engineering applications.

Two points are important to elucidate: the requirements on sampling design, and
the sample size requirements for reasonable power. Note that our experimental
detection procedure requires that samples be drawn from a Gaussian distribution;
this precludes factors which take values at fixed levels. Nonetheless, the
potential applications of this approach are myriad, as many physical systems
of practical interest feature continuous predictors.

Second, our results suggest that for low sample counts \((n<100)\), the detection
power may be unacceptably low \((<0.05)\). To achieve reasonable power, say
\(0.80\), our numerical experiments suggest that \(n>1000\) is necessary for these
detection procedures. While sample counts in the thousands are not uncommon for
computer experiments, such requirements are beyond a reasonable count for many
physical experiments. However, recent advances in microfluidics have enabled
kilohertz-throughput experiments which could easily reach our sampling
requirements.\cite{abate2010high} On the macro-scale, so-called \emph{cyber-physical
systems} enable the automated collection of data, such as \(1260\) unique cases of
pitching and heaving conditions of an airfoil.\cite{vanburen2112airfoils} Of
course, our immediate goal for future work is to reduce the sampling
requirements; the intent of the present article is to provide a lucid treatment
of detection fundamentals, and to illustrate these principles with minimal
assumptions.

Note that in this article, we make relatively modest assumptions on the
functional relation between response and predictors. Stein's lemma may be
regarded as implicitly utilizing the smoothness of the response. One could
potentially employ stronger assumptions to fruitful ends; namely, increasing
power.

Finally, we hope that both the analysis and procedures presented here prove
useful to further study. As Albrecht et al.\cite{albrecht2013} note, Dimensional
Analysis is less-studied in the statistics community, and certainly has more
untapped potential. We have found the analytical framework presented in Section
\ref{sec:org6f9d706} helpful in reasoning about lurking variables, and hope that
others find it similarly useful.
\section{Acknowledgments}
\label{sec:org025151b}
The authors would like to thank Jessica Hwang and Art Owen for useful comments
on the manuscript. The first author would like to thank Arman Sabbaghi for some
early comments on this work, and helpful pointers with regard to existing
literature. The first author is supported by the National Science Foundation
Graduate Research Fellowship under Grant No. DGE-114747. The second author was
supported by the NSF under Grants No. DMS-1407397 and DMS-1521145.

\bibliographystyle{plain}
\bibliography{man}

\section{Appendix}
\label{sec:orgf3f449b}
\begin{remark}[Unit Systems]\label{rmk:units}
There is historical precedent for changing unit systems; the 1875 Treaty of the
Metre established a standard unit of length based on a prototype metre, kept in
controlled conditions. This was redefined again in 1960 in terms of the
krypton-86 spectrum, to avoid the obvious issues of such a prototype
definition.\cite{britannica2017} The meter was redefined several times since,
and now is based on the distance light travels in a vacuum in $1/299,792,458$
seconds.\cite{thompson2008}
\end{remark}
\subsection{Unique non-dimensionalizing factor}
\label{sec:org8e1bcfd}
   Suppose we have some dimensional qoi \(\qoi\). We may form a dimensionless qoi
\(\pi\) by constructing a \emph{non-dimensionalizing} factor as a power-product of
the input quantities \(\v\var\). Such a non-dimensionalizing factor satisfies
\([\prod_{i=1}^{\nvar}\var_i^{u_i}] = [\qoi]\). Having found such a
\emph{non-dimensionalizing vector} \(\vu\in\R{\nvar}\), we may form

$$\pi=\qoi \prod_{i=1}^{\nvar}\var_i^{-u_i} = q\exp(-\vu^T\log(\v\var)).$$

Dimensional homogeneity demands that \(\vd(q)\in\cR(\mD)\), thus a
non-dimensionalizing factor always exists. However, an analyst may not be aware
of the full \(\v\var\), and may know only of the exposed variables \(\v\var_{\EX}\).
A non-dimensionalizing vector \(\vu_{\EX}\) may not exist for the exposed factors,
and is not necessarily unique. As Bridgman \cite{bridgman1922dimensional} notes,
one must have \(\vd(q)\in\cR(\mD_{\EX})\) in order for dimensional homogeneity to
hold. Below we prove existence and uniqueness of a particular \(\vu_{\EX}\) under
this condition.

\begin{theorem}[Existence of a unique non-dimensionalizing factor]\label{thm:non-dim}
If a physical relationship is dimensionally homogeneous in the full $\v\var$,
and some $\v\var_{\EX}$ are known with $\vd(q)\in\cR(\mD_{\EX})$, there exists a
unique non-dimensionalizing vector $\vu_{\EX}^*\in\R{\nvar_{\EX}}$ for the qoi
$\qoi$ that is orthogonal to the nullspace of $\mD_{\EX}$.
\end{theorem}

\begin{proof}[Proof of Theorem \ref{thm:non-dim}]
Since $\vd(\qoi)\in\cR(\mD_{\EX})$, we know a solution to
$\mD_{\EX}\vu_{\EX}=\vd(\qoi)$ exists. Denote $r_{\EX}=\text{Rank}(\mD_{\EX})$.
Employing the Rank-Nullity theorem, let
$\mV_{\EX}\in\mathbb{R}^{\nvar_{\EX}\times(\nvar_{\EX}-\rnk_{\EX})}$
be a basis for $\text{Null}(\mD_{\EX})$. Define the matrix

\begin{equation}
  \begin{aligned}
    \mM = \left[\begin{array}{c} \mD_{\EX} \\ \mV_{\EX}^T \end{array}\right],
  \end{aligned}
\end{equation}

\noindent and note that
$\mM\in\mathbb{R}^{(\ndim+\nvar_{\EX}-\rnk_{\EX})\times\nvar_{\EX}}$. Define the vector
$\vb\in\mathbb{R}^{(\ndim+\nvar_{\EX}-\rnk_{\EX})}$ via $\vb^T=[\vd(\qoi)^T,\vsym{0}^T]$,
where $\vsym{0}\in\mathbb{R}^{\nvar_{\EX}-\rnk_{\EX}}$. Then the solution to the
linear system $\mM\vu_{\EX}=\vb$ is a non-dimensionalizing vector for $\qoi$,
and is orthogonal to the nullspace of $\mD_{\EX}$.

Note that $\vd(\qoi)\in\cR(\mD_{\EX})$ and $\vsym{0}\in\cR(\mV_{\EX}^T)$, thus the
augmented matrix $[\mM|\vb]$ has the property
$\text{Rank}(\mM)=\text{Rank}([\mM|\vb])$. Note also that there are $\rnk_{\EX}$
independent rows in $\mD_{\EX}$ and $\nvar_{\EX}-\rnk_{\EX}$ independent rows in
$\mV_{\EX}^T$, with $\mD_{\EX}\mV_{\EX}=0$. Thus we have $\text{Rank}(\mM)=\nvar_{\EX}$.
By the \inlinelatex{Rouch\'{e}-Capelli} theorem, we know that a solution $\vu_{\EX}^*$ to
$\mM\vu_{\EX}=\vb$ exists and is unique.
\end{proof}
\end{document}

%% file: shared.tex
\input{zachs_macros}
\usepackage{amsthm}
\usepackage{float}    
\usepackage{booktabs} 
\usepackage{mathabx}  
\usepackage{cleveref}

\author{
  Zachary del Rosario\thanks{Department of Aeronautics and Astronautics,
    Stanford University, Stanford, CA
    (zdr@stanford.edu)}
  \and
  Minyong Lee\thanks{Department of Statistics, Stanford University, Stanford, CA
    (minyong@stanford.edu)}
  \and
  Gianluca Iaccarino\thanks{Department of Mechanical Engineering, Stanford, CA
    (jops@stanford.edu)}\footnotemark[3]
}

\usepackage{amsopn}

\newcommand{\var}{z} 
\newcommand{\tar}{x} 
\newcommand{\qoi}{q} 
\newcommand{\EX}{E}  
\newcommand{\LU}{L}  
\newcommand{\PI}{P}  
\newcommand{\nvar}{p} 
\newcommand{\ndim}{d} 
\newcommand{\rnk}{r} 
\newtheorem*{remark}{Remark}
\newtheorem{theorem}{Theorem}
\newcommand{\inlinelatex}[1]{#1}

%% file: zachs_macros.tex
\usepackage{amsmath}  
\usepackage{amsfonts} 
\usepackage{scalerel,stackengine} 
\usepackage{graphicx} 
\usepackage{caption}  
\usepackage{mathtools}
\usepackage{forest}   

\def\R#1{\mathbb{R}^{#1}}

\def\N#1{\mathbb{N}^{#1}}

\newcommand{\vsym}[1]{\boldsymbol{#1}}
\def\v#1{\vsym{#1}} 

\newcommand{\vb}{\boldsymbol{b}}

\newcommand{\vd}{\boldsymbol{d}}

\newcommand{\vg}{\boldsymbol{g}}

\newcommand{\vu}{\boldsymbol{u}}
\newcommand{\vv}{\boldsymbol{v}}
\newcommand{\vw}{\boldsymbol{w}}
\newcommand{\vx}{\boldsymbol{x}}

\newcommand{\msym}[1]{\boldsymbol{#1}}
\def\m#1{\msym{#1}} 

\newcommand{\mA}{\boldsymbol{A}}

\newcommand{\mD}{\boldsymbol{D}}

\newcommand{\mI}{\boldsymbol{I}}

\newcommand{\mM}{\boldsymbol{M}}

\newcommand{\mS}{\boldsymbol{S}}

\newcommand{\mV}{\boldsymbol{V}}
\newcommand{\mW}{\boldsymbol{W}}

\newcommand{\cR}{\mathcal{R}}

\def\E{\mathbb{E}} 

\def\P{\mathbb{P}} 

\def\V{\mathrm{V}} 

\stackMath
\newcommand\reallywidehat[1]{%
\savestack{\tmpbox}{\stretchto{%
  \scaleto{%
    \scalerel*[\widthof{\ensuremath{#1}}]{\kern-.6pt\bigwedge\kern-.6pt}%
    {\rule[-\textheight/2]{1ex}{\textheight}}
  }{\textheight}%
}{0.5ex}}%
\stackon[1pt]{#1}{\tmpbox}%
}

\newcommand{\ftree}[1]{%
\begin{forest}
for tree={
    font=\ttfamily,
    grow'=0,
    child anchor=west,
    parent anchor=south,
    anchor=west,
    calign=first,
    edge path={
      \noexpand\path [draw, \forestoption{edge}]
      (!u.south west) +(7.5pt,0) |- node[fill,inner sep=1.25pt] {} (.child anchor)\forestoption{edge label};
    },
    before typesetting nodes={
      if n=1
        {insert before={[,phantom]}}
        {}
    },
    fit=band,
    before computing xy={l=15pt},
  }%
  #1%
\end{forest}
}

\usepackage{numdef}   
\num\newcommand{\f12}{\frac{1}{2}}
\num
\num
\num
\num
\num
\num
\num

\num
\num
\num
\num
\num
\num
\num
\num

\num
\num
\num
\num
\num
\num
\num
\num

\num
\num
\num
\num
\num
\num
\num
\num

\num
\num
\num
\num
\num
\num
\num
\num

\num
\num
\num
\num
\num
\num
\num
\num

\num
\num
\num
\num
\num
\num
\num
\num

\num
\num
\num
\num
\num
\num
\num
\num

\num
\num
\num
\num
\num
\num
\num
\num

\num
\num
\num

\newcommand*{\mm}{%
  \leavevmode
  \hphantom{0}%
  \llap{%
    \settowidth{\dimen0 }{$0$}%
    \resizebox{1.1\dimen0 }{\height}{$-$}%
  }%
}

%% file: data/mom_pipe_raw.r_Ie1,2,3,4,5_Ip_sfac2.00.tex
\begin{tabular}{@{}llll@{}}
\toprule
 $n$ & \multicolumn{3}{c}{$\tau$} \\
 \cmidrule(lr){2-4}
 & $0.00$ & $100.00$ & $500.00$\\
\midrule
$100$  & $0.49$, $0.0750$ & $0.48$, $0.0736$ & $0.47$, $0.0693$\\
$200$  & $0.46$, $0.0747$ & $0.47$, $0.0757$ & $0.48$, $0.0731$\\
$400$  & $0.49$, $0.0791$ & $0.47$, $0.0782$ & $0.47$, $0.0710$\\
$800$  & $0.48$, $0.0779$ & $0.48$, $0.0813$ & $0.47$, $0.0715$\\
$1600$ & $0.49$, $0.0853$ & $0.50$, $0.0788$ & $0.48$, $0.0756$\\
$3200$ & $0.48$, $0.0830$ & $0.50$, $0.0814$ & $0.48$, $0.0757$\\
$6400$ & $0.51$, $0.0828$ & $0.49$, $0.0801$ & $0.49$, $0.0794$\\
\bottomrule
\end{tabular}

%% file: data/mom_two_fluid_raw.r_Ie1,2,3,4,5,6,7_Ip_sfac6.00.tex
\begin{tabular}{@{}llll@{}}
\toprule
 $n$ & \multicolumn{3}{c}{$\tau$} \\
 \cmidrule(lr){2-4}
 & $0.00$ & $0.50$ & $1.00$ \\
\midrule
$100$  & $0.48$, $0.0743$ & $0.50$, $0.0820$ & $0.51$, $0.0803$\\
$200$  & $0.48$, $0.0832$ & $0.48$, $0.0793$ & $0.49$, $0.0825$\\
$400$  & $0.47$, $0.0773$ & $0.50$, $0.0843$ & $0.51$, $0.0838$\\
$800$  & $0.50$, $0.0815$ & $0.49$, $0.0831$ & $0.49$, $0.0801$\\
$1600$ & $0.51$, $0.0807$ & $0.49$, $0.0839$ & $0.50$, $0.0807$\\
$3200$ & $0.49$, $0.0799$ & $0.49$, $0.0910$ & $0.50$, $0.0832$\\
$6400$ & $0.51$, $0.0835$ & $0.51$, $0.0829$ & $0.51$, $0.0789$\\
\bottomrule
\end{tabular}